\documentclass[12pt, a4paper, reqno]{amsart}
\usepackage{mathrsfs}
\usepackage{bbm}
\usepackage{geometry}
\usepackage{amsmath,amscd,amssymb,latexsym}
\usepackage[all]{xy}
\usepackage{color}
\usepackage[section]{placeins}
\input xypic
\usepackage{enumerate}
\DeclareMathVersion{can}
\DeclareMathAlphabet{\can}{OT1}{cmss}{m}{n}
\newcommand{\T}{^{\top}}

\numberwithin{equation}{section}

\numberwithin{equation}{section}

\newtheorem{theorem}{Theorem}[section]
\newtheorem{lemma}{Lemma}[section]
\newtheorem{corollary}{Corollary}[section]
\theoremstyle{definition}
\newtheorem{remark}{Remark}[section]
\theoremstyle{definition}
\newtheorem{definition}{Definition}[section]
\theoremstyle{definition}
\newtheorem{example}{Example}[section]
\begin{document}
\title[  Additive non-RS MDS codes and ACD codes from
twisted Reed-Solomon codes]{New families of additive non-RS MDS codes and ACD codes from
twisted Reed-Solomon codes}
\author [Li] {Xia Li}
\thanks{Xia~Li is with College of Science, Nanjing University of Posts and Telecommunications, Hubei Key Laboratory of Applied Mathematics (Hubei University), Nanjing,  211100, P. R. China (e-mail: lixia4675601@163.com).}

\author[Wang]{Yu Wang}
\thanks{Yu Wang is with College of Science, Nanjing University of Posts and Telecommunications, Nanjing,  211100, P. R. China (e-mail: 2225075357@qq.com).}

\author[Wu]{Yansheng Wu}
\thanks{Yansheng Wu is with School of Computer Science, Nanjing University of Posts and Telecommunications, Nanjing,  211100, P. R. China (e-mail: yanshengwu@njupt.edu.cn).}

\thanks{The work of Xia Li was supported by  National Natural Science Foundation of China (No.
12301668) and Open Foundation of Hubei Key Laboratory of Applied Mathematics (Hubei University). The work of Yansheng Wu was supported in part by the National Natural
Science Foundation of China under Grant 62372247 and Grant 12101326,
 and in part by China Postdoctoral Science Foundation,
under Grant 2023M740958. }
\begin{abstract}
Additive codes generalize linear codes and play an important role in the construction of
quantum codes. As generalizations of linear MDS codes, additive MDS codes have received much research attention. As an important subclass of additive codes, ACD codes are widely used for quantum error correction. 
In this paper, we study single-twist additive generalized twisted Reed-Solomon (GTRS) codes over finite fields.
We derive necessary and sufficient conditions for these codes to be MDS for arbitrary hook and twist parameters.
Then we employ Schur squares to prove that such additive MDS codes are not monomially equivalent to additive Reed-Solomon codes.
We also identify a class of additive GTRS codes whose duals remain additive GTRS codes, and show that, for all sufficiently large prime powers $q$ whose characteristic avoids a finite set, such GTRS codes yield ACD codes.
\end{abstract}
\maketitle
\section{Introduction}Maximum distance separable (MDS) codes are optimal block codes that satisfy the equality condition of the Singleton bound.
The Hamming distance quantifies the error-detection and error-correction performance of error-correcting codes.
MDS codes meet the Singleton bound with equality and attain the largest possible minimum distance for fixed length and code size.
Singleton \cite{8} originally proposed and investigated MDS codes, motivated by their utility for constructing large constant-weight binary codes with considerable Hamming distances.
Reed and Solomon \cite{9} developed generalized Reed-Solomon (GRS) codes, the most extensively researched linear MDS code family to date.
Relying on strong burst error correction ability, GRS codes guarantee stable data transmission for optical discs and digital audio tapes.
In addition, they serve as fundamental building blocks for DNA error correction codes \cite{10,11} and locally recoverable codes deployed in distributed storage systems \cite{12,13}.
Besides GRS codes, there are other well-known constructions of linear MDS codes with the help of n-arcs in projective geometry \cite{14,15} and circulant Cauchy matrices \cite{16}.
MDS codes find practical uses in network coding, cryptography, data storage and quantum mechanics and they are closely connected with $n$-arcs and orthogonal arrays, motivating abundant studies on their constructions \cite{5,17,18,19,20}.

Additive codes generalize linear codes over finite fields, maintaining favorable algebraic structures despite being nonlinear objects and can be completely described in terms of a basis.
Calderbank et al. \cite{20} first investigated $\mathbb{F}_2$-additive codes over $\mathbb{F}_4$ as well as their duals with respect to the trace bilinear form, and further adopted these codes to construct quantum stabilizer codes. Bierbrauer and Edel later generalized the above theoretical framework to arbitrary finite fields \cite{21}.
Recently, Shi et al. \cite{22} pointed out that designing orthogonal direct-sum masking schemes only requires the additivity and complementarity of codes rather than linearity.
Such masking schemes are widely adopted to defend sensitive information against side-channel attacks (SCA) and fault injection attacks (FIA).
Accordingly, additive codes possessing complementary duals (ACD codes) can serve as countermeasures against both passive and active side-channel analysis in embedded cryptosystems.
Driven by this practical motivation, Shi et al. \cite{22} investigated $\mathbb{F}_2$-additive complementary dual codes over $\mathbb{F}_4$, and constructed $\mathbb{F}_4$-valued ACD codes from binary codes via the standard trace bilinear form as well as the Hermitian trace bilinear form.
In another recent work, Choi et al. \cite{23} provided methods to construct ACD codes over finite fields. 
They also listed some ACD codes with good parameters over $\mathbb{F}_4$, $\mathbb{F}_8$ and $\mathbb{F}_9$, and identified some MDS ACD codes among these codes.

In a recent work, Beelen et al. \cite{24} introduced and studied twisted Reed-Solomon (TRS) codes with one twist as a natural generalization of RS codes and showed that these codes are not MDS in general, and characterized two MDS subclasses of single-twist TRS codes. 
In another study \cite{25}, the authors pointed out that the dual codes of ordinary TRS codes usually do not belong to the TRS family, but a special type of TRS codes is closed under the dual operation; additionally, a kind of TRS codes can resist Sidelnikov-Shestakov and Wieschebrink attacks against the McEliece cryptosystem.
Subsequently, Beelen et al. \cite{5} generalized single-twist TRS codes to TRS codes with $\ell$ twists, constructed multiple MDS TRS subfamilies, and proved that plenty of TRS codes are not monomially equivalent to classical RS codes. 
Furthermore, Liu and Liu \cite{19} put forward construction approaches for MDS LCD codes over finite fields by virtue of generalized twisted Reed-Solomon (GTRS) codes equipped with $\ell$ twists.

Huffman extended the Singleton-type bound for additive codes in subsequent research in \cite{3}. 
As optimal codes achieving this bound, additive MDS codes correspond to geometric structures such as pseudo-arcs \cite{26} and serve as fundamental building blocks for quantum stabilizer codes \cite{27}. Ball et al. \cite{28} proved that any $\mathbb{F}_q$-additive MDS code with length $n$ and additive dimension $2k$ over $\mathbb{F}_{q^2}$ is linear over $\mathbb{F}_{q^2}$ under certain premises.
In recent work, Yadav and Sharma \cite{2} investigated new additive codes from generalized twisted Reed-Solomon codes (GTRS codes) and they obtained the necessary and sufficient conditions for the code to be MDS when $t=1$ and $0 \leq h \leq k-1$.
Jiayu Ma et al. \cite{1} derived the necessary and sufficient conditions for the code to be MDS when $t=2$ and $h = k-1$.
Thus, this paper focuses on constructing additive MDS codes via additive GTRS codes for general $h$ and $t$.

This paper is organized as follows.
In Section 2, we introduce the necessary foundational notions to derive our main results, which consist of additive codes, their Singleton bound, additive GRS code, additive GTRS code and the Schur product.
In Section 3, we construct a new class of additive non-RS MDS codes from twisted Reed-Solomon codes.
We also identify a class of additive GTRS codes whose duals remain additive GTRS codes.
In Section 4, we provide a construction for ACD codes based on additive GTRS codes.
In Section 5, we present a brief conclusion.

\section{Preliminary}
In this section, we list basic definitions needed for proving our main results.
Let $q$ be a prime power, $r\geq 2$ an integer and $n$ denote a positive integer. 
The set $\mathbb{F}_{q^r}^n$ consisting of all $n$-dimensional tuples with entries taken from $\mathbb{F}_{q^r}$ forms a vector space of dimension $nr$ over the ground field $\mathbb{F}_q$. 
The vector space operations here refer to component-wise addition together with scalar multiplication where scalars are selected from $\mathbb{F}_q$.
Two additive codes of length \(n\) over \(\mathbb{F}_{q^r}\) are said to be monomially equivalent if a generator matrix of one code can be obtained from the generator matrix of the other code by post-multiplying it with an \(n \times n\) monomial matrix over \(\mathbb{F}_{q^r}\).
It is a standard property that monomially equivalent codes possess identical triple parameters $n,k,d$.

An $\mathbb{F}_q$-additive code $\mathcal{C}$ of length $n$ over the extension field $\mathbb{F}_{q^r}$ is defined as an $\mathbb{F}_q$-linear subspace embedded inside the ambient space $\mathbb{F}_{q^r}^n$.
Throughout the subsequent discussion, we write ``additive code'' for \(\mathbb{F}_q\)-additive code over \(\mathbb{F}_{q^r}\).
We write $\dim_{\mathbb{F}_q}(\mathcal{C})$ to represent the dimension of the code subspace when viewed as a linear space over $\mathbb{F}_q$.
Correspondingly, a generator matrix for $\mathcal{C}$ is a matrix with entries in $\mathbb{F}_{q^r}$, whose rows form an \(\mathbb{F}_q\)-basis of \(C\).
For any nonzero codeword $\boldsymbol{c}\in\mathcal{C}$, $w_H(\boldsymbol{c})$ stands for its Hamming weight, namely the number of nonzero coordinates contained in the vector $\boldsymbol{c}$. Based on this definition of weight, the minimum Hamming distance of code $\mathcal{C}$ is defined as
\[
d_H(\mathcal{C}) = \min\left\{ w_H(\boldsymbol{c}) \mid \boldsymbol{c}\in\mathcal{C},\; \boldsymbol{c}\neq \boldsymbol{0} \right\}.
\]
If an additive code satisfies $\dim_{\mathbb{F}_q}\mathcal{C}=k$ and its minimum Hamming distance equals $d$, we call such a code an additive $[n,k,d]$ code over $\mathbb{F}_{q^r}$.

The trace inner product on $\mathbb{F}_{q^r}^n$ is defined as the mapping
\[
\langle \boldsymbol{u},\boldsymbol{v}\rangle = \sum_{i=1}^n \operatorname{Tr}(u_i v_i)
\]
for arbitrary vectors $\boldsymbol{u}=(u_1,u_2,\dots,u_n),\boldsymbol{v}=(v_1,v_2,\dots,v_n) \in \mathbb{F}_{q^r}^n$, where $\operatorname{Tr}$ stands for the standard trace mapping from $\mathbb{F}_{q^r}$ to $\mathbb{F}_q$. According to Lemma 1 in Huffman \cite{3}, this trace inner product is a symmetric and non-degenerate bilinear form over $\mathbb{F}_{q^r}^n$.
Given an additive code $\mathcal{C}$ with length $n$ over $\mathbb{F}_{q^r}$, its dual code $\mathcal{C}^\perp$ is defined by
\[
\mathcal{C}^\perp = \big\{\boldsymbol{u} \in \mathbb{F}_{q^r}^n \,\big|\, \langle \boldsymbol{u},\boldsymbol{c}\rangle = 0,\ \forall \boldsymbol{c}\in \mathcal{C}\big\}.
\]
It can be verified that $\mathcal{C}^\perp$ is also an $\mathbb{F}_q$-linear subspace of $\mathbb{F}_{q^r}^n$, so it forms another additive code. The generator matrix of $\mathcal{C}^\perp$ is referred to as the parity-check matrix of the original code $\mathcal{C}$. From Proposition 2.4 in reference \cite{4}, we have
\[
\dim_{\mathbb{F}_q}(\mathcal{C}) + \dim_{\mathbb{F}_q}(\mathcal{C}^\perp) = nr.
\]
We say an additive code $\mathcal{C}$ is an additive complementary dual (ACD) code when the intersection of $\mathcal{C}$ and its dual space only contains the zero vector, i.e., $\mathcal{C} \cap \mathcal{C}^\perp = \{\boldsymbol{0}\}$.

We then present the Singleton bound for additive codes from \cite{3}.
\begin{lemma}
\label{lemma2.1}
Suppose $\mathcal{C}$ is an additive-$[n,k,d]$ code defined over $\mathbb{F}_{q^r}$, then the inequality
\[
d \leq n - \left\lceil \frac{k}{r} \right\rceil + 1
\]
holds, where the symbol $\lceil \cdot \rceil$ stands for the ceiling function.
\end{lemma}

We name $\mathcal{C}$ an additive MDS code when its minimum distance satisfies $d = n - \left\lceil \dfrac{k}{r} \right\rceil + 1$.
Correspondingly, $\mathcal{C}$ is referred to as an additive almost MDS code provided $d = n - \left\lceil \dfrac{k}{r} \right\rceil$.

We fix a notation for element-wise trace operation on matrices for the whole paper.
Given an $m \times \ell$ matrix $A$ defined over $\mathbb{F}_{q^r}$, let its $(i,j)$-th entry be $a_{i,j}$. We use $\operatorname{Tr}(A)$ to represent an $m\times\ell$ matrix with entries in $\mathbb{F}_q$, where the $(i,j)$-th entry equals $\operatorname{Tr}(a_{i,j})$ for all $1 \le i \le m$, $1 \le j \le \ell$. It should be emphasized that this matrix trace operation is different from the standard matrix trace in linear algebra.

Suppose an additive code $\mathcal{C}$ over $\mathbb{F}_{q^r}$ has generator matrix $G$ and parity-check matrix $H$.
The identity $\operatorname{Tr}(GH\T)=\mathbf{0}$ holds necessarily. 
We write $A\T$ for the transpose of matrix $A$.

\begin{lemma}
\label{lemma2.2}(\cite{2})
Let $C$ be an additive code of length $n$ over $\mathbb{F}_{q^r}$ with a generator matrix $G$. Then the code $C$ is ACD if and only if the matrix $\operatorname{Tr}\bigl(GG\T\bigr)$ is non-singular.
\end{lemma}

To define additive GRS codes, let \(n, k\) and \(r \ge 2\) be integers satisfying \(1 \le k \le nr\). Let \(\boldsymbol{\alpha} = (\alpha_1, \alpha_2, \dots, \alpha_n) \in \mathbb{F}_{q^r}^n\) and \(\boldsymbol{v} = (v_1, v_2, \dots, v_n) \in (\mathbb{F}_{q^r}^*)^n\), where \(\alpha_1, \alpha_2, \dots, \alpha_n\) are distinct. 
Consider the evaluation map \(\mathcal{E}_{\boldsymbol{\alpha}, \boldsymbol{v}}: \mathbb{F}_q[x]_{<k} \to \mathbb{F}_{q^r}^n\), defined as
\[
\mathcal{E}_{\boldsymbol{\alpha}, \boldsymbol{v}}(f(x)) = (v_1f(\alpha_1), v_2f(\alpha_2), \dots, v_nf(\alpha_n)) \quad \text{for all } f(x) \in \mathbb{F}_q[x]_{<k}.
\]
One can easily see that \(\mathcal{E}_{\boldsymbol{\alpha}, \boldsymbol{v}}\) is an \(\mathbb{F}_q\)-linear vector space homomorphism, and hence its image \(\mathcal{E}_{\boldsymbol{\alpha}, \boldsymbol{v}}(\mathbb{F}_q[x]_{<k})\) is an \(\mathbb{F}_q\)-linear subspace of \(\mathbb{F}_{q^r}^n\). 
The additive generalized Reed-Solomon (additive GRS) code of length \(n\) over \(\mathbb{F}_{q^r}\) with evaluation points \(\alpha_1, \alpha_2, \dots, \alpha_n\) and column multipliers \(v_1, v_2, \dots, v_n\) is defined as
\[
\mathcal{ARS}_k(\boldsymbol{\alpha}, \boldsymbol{v}) = \mathcal{E}_{\boldsymbol{\alpha}, \boldsymbol{v}}(\mathbb{F}_q[x]_{<k}) = \{(v_1f(\alpha_1), v_2f(\alpha_2), \dots, v_nf(\alpha_n)) \mid f(x) \in \mathbb{F}_q[x]_{<k}\}.
\]
In particular, when \(\boldsymbol{v} = (1, 1, \dots, 1)\), this code is referred to as an additive RS code, i.e., \(\mathcal{ARS}_k(\boldsymbol{\alpha})\). 
Note that the evaluation map \(\mathcal{E}_{\boldsymbol{\alpha}, \boldsymbol{v}}\) is not always injective under the mere assumption that \(\alpha_i\) are distinct. When the map \(\mathcal{E}_{\boldsymbol{\alpha}, \boldsymbol{v}}\) is injective, the code \(\mathcal{ARS}_k(\boldsymbol{\alpha}, \boldsymbol{v})\) has \(\mathbb{F}_q\)-dimension exactly \(k\). 

In this paper the injectivity is never an issue: whenever we use the code $\mathcal{ARS}_k(\boldsymbol{\alpha},\boldsymbol{v})$ or the code $\mathscr{T}_k(\boldsymbol{\alpha},\boldsymbol{v},\mathbf{t},\mathbf{h},\boldsymbol{\eta})$, the evaluation points $\alpha_1,\dots,\alpha_n$ lie in pairwise distinct full-degree Frobenius orbits and $k\le nr$, so that the $nr$ elements $\alpha_i^{q^j}$ are pairwise distinct and a nonzero polynomial of degree less than $nr$ over $\mathbb{F}_q$ cannot vanish at all of them.

Linear GTRS codes are well-known constructions for non-RS MDS codes. In parallel with this framework, we introduce the definition of additive GTRS codes.
Let $\ell$ be a positive integer, pick
\[
\mathbf{t}=(t_1,t_2,\dots,t_\ell) \in \{1,2,\dots,nr-k\}^\ell,\quad
\mathbf{h}=(h_1,h_2,\dots,h_\ell) \in \{0,1,\dots,k-1\}^\ell,
\]
where every pair $(h_j,t_j)$ is mutually distinct for $1\le j\le \ell$. Let $\boldsymbol{\eta}=(\eta_1,\eta_2,\dots,\eta_\ell)\in\mathbb{F}_q^\ell$.
Here $\ell$ represents the total number of twist terms, $\mathbf{t}$ is the twist parameter vector, $\mathbf{h}$ denotes the hook vector, and $\boldsymbol{\eta}$ corresponds to the coefficient vector of twisted monomials.
Define this polynomial subspace as
\[
\mathscr{P}_{n,k}(\mathbf{t},\mathbf{h},\boldsymbol{\eta}) =
\left\{
\sum_{i=0}^{k-1} a_i x^i + \sum_{j=1}^{\ell} \eta_j a_{h_j} x^{k-1+t_j} \,\bigg|\, a_i\in\mathbb{F}_q
\right\}
\subseteq \mathbb{F}_q[x].
\]

By Lemma 1 of Beelen  \cite{5}, we see that the set $\mathscr{P}_{n,k}(\mathbf{t},\mathbf{h},\boldsymbol{\eta})$ is a $k$-dimensional subspace of $\mathbb{F}_q[x]$ over $\mathbb{F}_q$ with a basis set $\{p_0(x),p_1(x),\dots,p_{k-1}(x)\}$, where
\[
p_i(x) = x^i + \sum_{\substack{j=1\\ h_j=i}}^\ell \eta_j x^{k-1+t_j} \quad \text{for } 0 \le i \le k-1.
\]

Based on the above polynomial set, additive GTRS codes of length $n$ over $\mathbb{F}_{q^r}$ with $\ell$ twists and evaluation points $\boldsymbol{\alpha}$ are defined by
\[
\mathscr{T}_k(\boldsymbol{\alpha},\boldsymbol{v},\mathbf{t},\mathbf{h},\boldsymbol{\eta})
= \big\{ \big(v_1 f(\alpha_1),v_2 f(\alpha_2),\dots,v_n f(\alpha_n)\big) \,\big|\, f(x) \in \mathscr{P}_{n,k}(\mathbf{t},\mathbf{h},\boldsymbol{\eta}) \big\}.
\]
If we take $\boldsymbol{v}=(1,1,\dots,1)$, the resulting code is named an additive TRS code and marked with the notation $\mathscr{T}_k(\boldsymbol{\alpha},\mathbf{t},\mathbf{h},\boldsymbol{\eta})$.

\begin{definition}
\label{definition2.1}
Given two arbitrary vectors $\boldsymbol{x} = (x_1,x_2,\dots,x_n)$ and $\boldsymbol{y} = (y_1,y_2,\dots,y_n)$ belonging to the vector space $\mathbb{F}_{q^r}^n$, their Schur product is formulated component-wise as:
\[
\boldsymbol{x} \star \boldsymbol{y} = \left(x_1 y_1,\, x_2 y_2,\, \dots,\, x_n y_n\right).
\]
\end{definition}
Consider two additive codes $\mathcal{C}_1,\mathcal{C}_2$ which are subsets embedded in $\mathbb{F}_{q^r}^n$. The Schur product space generated by $\mathcal{C}_1$ and $\mathcal{C}_2$ is constructed by taking the $\mathbb{F}_q$-linear span of all component-wise Schur products computed from codewords picked separately from $\mathcal{C}_1$ and $\mathcal{C}_2$, whose mathematical expression is written as
\[
\mathcal{C}_1 \star \mathcal{C}_2 := \operatorname{span}_{\mathbb{F}_q}\big\{\,\boldsymbol{c}_1 \star \boldsymbol{c}_2 \,\big|\, \boldsymbol{c}_1 \in \mathcal{C}_1,\,\boldsymbol{c}_2 \in \mathcal{C}_2\big\}.
\]

\begin{lemma}
\label{lemma2.3}(\cite{2})
\begin{enumerate}
  \item  If $\mathcal{C}$ is an additive code of length $n$ and dimension $k$ over $\mathbb{F}_{q^r}$, then the dimension of its Schur square $\mathcal{C}^2$ satisfies
$
\dim_{\mathbb{F}_q}(\mathcal{C}^2) \leq \min\left\{nr,\,\frac{k(k+1)}{2}\right\}.
$
  \item If $\mathcal{C}$ and $\mathcal{D}$ are two monomially equivalent additive codes over $\mathbb{F}_{q^r}$, then $\dim_{\mathbb{F}_q}(\mathcal{C}^2) = \dim_{\mathbb{F}_q}(\mathcal{D}^2)$.
\end{enumerate}

\end{lemma}

\begin{lemma}
\label{lemma2.4}(\cite{2})
Let $n,k$ and $r \geq 2$ be positive integers satisfying $1 \leq k \leq nr$. Let $\boldsymbol{\alpha} = (\alpha_1, \alpha_2, \dots, \alpha_n) \in \mathbb{F}_{q^r}^n$, where no two elements among $\alpha_1, \alpha_2, \dots, \alpha_n$ form a conjugate pair over $\mathbb{F}_q$ and each $\alpha_i$ has exactly $r$ distinct conjugates over $\mathbb{F}_q$. We have
\[
\dim_{\mathbb{F}_q}\big(\mathcal{ARS}_{k}(\boldsymbol{\alpha},1)^2\big) = \min\left\{nr, 2k - 1\right\}.
\]
\end{lemma}

\begin{lemma}
\label{lemma2.5}
Let $\mathcal{C} \subseteq \mathbb{F}_{q^r}^n$ be an additive code satisfying $\dim_{\mathbb{F}_q}(\mathcal{C}) \leq \dfrac{nr}{2}$. If
$
\dim_{\mathbb{F}_q}\big(\mathcal{C}^2\big) \neq 2\dim_{\mathbb{F}_q}(\mathcal{C}) - 1,
$
then $\mathcal{C}$ is not monomially equivalent to the additive RS code $\mathcal{ARS}_k(\boldsymbol{\alpha},1)$, where no two elements among $\alpha_1, \alpha_2, \dots, \alpha_n$ form a conjugate pair over $\mathbb{F}_q$ and each $\alpha_i$ has exactly $r$ distinct conjugates over $\mathbb{F}_q$.

\begin{proof}
Assume to the contrary that $\mathcal{C}$ is monomially equivalent to the additive RS code $\mathcal{ARS}_k(\boldsymbol{\alpha},1)$. Let $k=\dim_{\mathbb{F}_q}(\mathcal{C})$.
From the known condition $k\leq \dfrac{nr}{2}$, we get $2k-1\leq nr-1$. According to lemma \ref{lemma2.4},
\[
\dim_{\mathbb{F}_q}\big(\mathcal{ARS}_k(\boldsymbol{\alpha},1)^2\big)=2k-1.
\]
Lemma \ref{lemma2.3} part (b) tells us that monomial equivalent codes have the same dimension for their Schur squares. Hence $\dim_{\mathbb{F}_q}(\mathcal{C}^2)=2k-1$.
However, our given condition is $\dim_{\mathbb{F}_q}(\mathcal{C}^2)\neq 2k-1$. This leads to a contradiction.
Thus, $\mathcal{C}$ is not monomially equivalent to $\mathcal{ARS}_k(\boldsymbol{\alpha},1)$.
\end{proof}
\end{lemma}

Suppose $I = \{i_1,i_2,\dots,i_\ell\} \subseteq \mathbb{Z}$ is a finite set of integers ordered as $i_1 < i_2 < \dots < i_\ell$. We call $I$ an arithmetic progression if the gap $i_j - i_{j-1}$ between adjacent numbers remains unchanged for every integer $j$ satisfying $2 \le j \le \ell$. Otherwise, $I$ is called a non-arithmetic progression. Notice that any set with one or two elements is automatically an arithmetic progression.

Let $I$ be an integer set with $\ell$ elements. Its sumset is defined by
\[
I + I = \{a + b \mid a,b \in I\}.
\]

\begin{lemma}
\label{lemma2.6}(\cite{6})
Let $I$ be a set with $l$ integers. Then $|I+I| \geq 2l-1$. Moreover, if $I$ is a set with $l$ integers such that $|I+I| = 2l-1$, then $I$ is an arithmetic progression.
It implies that the cardinality of $I+I$ is no less than $2l$ for any non-arithmetic progression $I$.
\end{lemma}

\begin{definition}
\label{definition2.2}
Let $n$, $k$ and $r \ge 2$ be integers with $1 \le k < nr$. Let $\ell \ge 1$ be an arbitrary integer, and take
$\boldsymbol{t} = (t_1,t_2,\dots,t_\ell) \in \{1,2,\dots,nr-k\}^\ell$,
$\boldsymbol{h} = (h_1,h_2,\dots,h_\ell) \in \{0,1,2,\dots,k-1\}^\ell$,
such that all pairs $(h_1,t_1),(h_2,t_2),\dots,(h_\ell,t_\ell)$ are mutually distinct.
Let $\boldsymbol{\eta} = (\eta_1,\eta_2,\dots,\eta_\ell) \in (\mathbb{F}_q^*)^\ell$.
Define the degree set
\[
\mathcal{D}_{n,k}(\boldsymbol{t},\boldsymbol{h},\boldsymbol{\eta}) = \bigl\{\deg(p_i(x)) : 0 \le i \le k-1\bigr\}.
\]
Then
\begin{align*}
\mathcal{D}_{n,k}(\boldsymbol{t},\boldsymbol{h},\boldsymbol{\eta})
&= \bigl(\{0,1,\dots,k-1\} \setminus \{h_j : 1 \le j \le \ell\}\bigr) \notag \\
&\quad \cup \left\{k-1+\max_{\substack{1\le s\le\ell\\ h_s=i}} \{t_s\} : i \in \{h_a : 1 \le a \le \ell\} \right\}.
\end{align*}

We further introduce the set
\[
\mathcal{F}_{n,k}(\boldsymbol{t},\boldsymbol{h},\boldsymbol{\eta})
= \bigl\{d_1 + d_2 : d_1,d_2 \in \mathcal{D}_{n,k}(\boldsymbol{t},\boldsymbol{h},\boldsymbol{\eta}),\ d_1 + d_2 < nr\bigr\}.
\]
\end{definition}

\begin{lemma}(\cite{2})
\label{lemma2.7}
Let $n$, $k$ and $r \ge 2$ be positive integers satisfying $1 \le k < nr$.
Let $\boldsymbol{\alpha} = (\alpha_1,\alpha_2,\dots,\alpha_n) \in \mathbb{F}_{q^r}^n$, where no two elements among $\alpha_1,\alpha_2,\dots,\alpha_n$ form a conjugate pair over $\mathbb{F}_q$ and each $\alpha_i$ has exactly $r$ distinct conjugates over $\mathbb{F}_q$. Let
$\boldsymbol{t} = (t_1,t_2,\dots,t_\ell) \in \{1,2,\dots,nr-k\}^\ell$ and $\boldsymbol{h} = (h_1,h_2,\dots,h_\ell) \in \{0,1,2,\dots,k-1\}^\ell$ be such that the pairs $(h_1,t_1),(h_2,t_2),\dots,(h_\ell,t_\ell)$ are distinct, and let $\boldsymbol{\eta} = (\eta_1,\eta_2,\dots,\eta_\ell) \in (\mathbb{F}_q^*)^\ell$.
Then
\[
\dim_{\mathbb{F}_q}\bigl(\mathscr{T}_k(\boldsymbol{\alpha}, \boldsymbol{t},\boldsymbol{h},\boldsymbol{\eta})^2\bigr) \ge \bigl|\mathcal{F}_{n,k}(\boldsymbol{t},\boldsymbol{h},\boldsymbol{\eta})\bigr|.
\]
\end{lemma}

Throughout this paper, we fix $\boldsymbol{v}=(1,1,\dots,1)$ and $\ell=1$ by default.

Let \(\mathfrak{B}\) be the complete set of representatives of all full-degree Frobenius orbits in \(\mathbb{F}_{q^r}\) over \(\mathbb{F}_q\). That is, any two distinct elements in \(\mathfrak{B}\) are not conjugate over \(\mathbb{F}_q\), and each element in \(\mathfrak{B}\) has exactly \(r\) distinct conjugates over \(\mathbb{F}_q\).
The cardinality formula of such a set $\mathfrak{B}$ has been derived in reference \cite{7}, written as
\[
|\mathfrak{B}| = \frac{1}{r}\sum_{d\mid r}\mu(d) q^{\frac{r}{d}}.
\]
Here \(\mu(\cdot)\) stands for the M\"obius function defined on positive integers, whose explicit rule is given below:
\[
\mu(d)=
\begin{cases}
1, & d=1,\\
(-1)^t, & d \text{ is a product of } t \text{ distinct prime factors},\\
0, & d \text{ contains a squared prime divisor}.
\end{cases}
\]
\section{Additive MDS codes from additive GTRS codes over finite fields}

In this section, for $\ell=1$, $1\le t\le nr-k$ and $0\le h\le k-1$, we prove that the code $\mathscr{T}_k(\boldsymbol{\alpha},t,h,\eta)$ is not monomially equivalent to an additive RS code of the same length and dimension.
Then we present a construction of additive MDS codes from GTRS codes in this range of parameters.

\begin{theorem}
\label{theorem3.1}
Let $q$ be a prime power and $n,k,t$ and $r \geq 2$ be positive integers such that $n \le |\mathfrak{B}|$, $3 \leq k < \frac{nr}{2}+1-t$, $1 \leq t \leq nr-k$, $h \in \{0,1,2,\dots,k-1\}$, $\eta \in \mathbb{F}_q^*$.
Let $\boldsymbol{\alpha}=(\alpha_1,\alpha_2,\dots,\alpha_n)\in \mathbb{F}_{q^r}^n$, where no two elements among $\alpha_1,\alpha_2,\dots,\alpha_n$ form a conjugate pair over $\mathbb{F}_q$ and each $\alpha_i$ has exactly $r$ distinct conjugates over $\mathbb{F}_q$.
Then we have $\dim_{\mathbb{F}_q}\bigl(\mathscr{T}_k(\boldsymbol{\alpha}, t, h, \eta)^2\bigr) \geq 2k$.
Furthermore, the code $\mathscr{T}_k(\boldsymbol{\alpha}, t, h, \eta)$  is not monomially equivalent to an additive RS code $\mathcal{ARS}_k(\boldsymbol{\alpha},1)$.

\begin{proof}
The case when $t=1$ has been proved in Theorem 6.2 in \cite{2}.
Next, when $t \geq 2$, we have $p_i(x) = x^i,\ i \in \{0,1,2,\dots,k-1\}\setminus \{h\}
\quad \text{and} \quad
p_h(x) = x^h + \eta x^{k-1+t}$.
The degree set $\mathcal{D}_{n,k}(\boldsymbol{t},\boldsymbol{h},\boldsymbol{\eta})$ is $\{0,1,2,\dots,k-1\}\setminus \{h\} \cup \{k-1+t\}$.
Since \(k < \frac{nr}{2} + 1 - t\), we have \(2(k-1+t) < nr\), which guarantees that all degrees in the Schur square are strictly less than \(nr\).
When \((k,t,h)=(3,2,1)\), then $\mathcal D_{n,k}(t,h,\eta)=\{0,2,4\}$ is an arithmetic progression, then we handle this exceptional case directly. 
The basis polynomials are $p_0(x)=1, p_1(x)=x+\eta x^4, p_2(x)=x^2$.
The six Schur products
\begin{gather*}
p_0 * p_0 = 1, \qquad p_0 * p_1 = x + \eta x^4, \qquad p_0 * p_2 = x^2, \\
p_1 * p_1 = x^2 + 2\eta x^5 + \eta^2x^8, \qquad p_1 * p_2 = x^3 + \eta x^6, \qquad p_2 * p_2 = x^4,
\end{gather*}
are linearly independent over \(\mathbb{F}_q\).
Hence $\dim_{\mathbb{F}_q}\bigl(\mathscr{T}_k(\boldsymbol{\alpha}, t, h, \eta)^2\bigr) \ge 6=2k$.
In all remaining cases, the degree set is not an arithmetic progression and we have $|\mathcal{F}_{n,k}(\boldsymbol{t},\boldsymbol{h},\boldsymbol{\eta})| \geq 2k$ by Lemma \ref{lemma2.6}.
By Lemma \ref{lemma2.7}, $\dim_{\mathbb{F}_q}\bigl(\mathscr{T}_k(\boldsymbol{\alpha}, t, h, \eta)^2\bigr) \geq 2k$.
Then by Lemma \ref{lemma2.5}, we derive that the code $\mathscr{T}_k(\boldsymbol{\alpha}, t, h, \eta)$  is not monomially equivalent to an additive RS code $\mathcal{ARS}_k(\boldsymbol{\alpha},1)$.
\end{proof}
\end{theorem}

\begin{lemma}(\cite{2})
\label{lemma3.1}
Let  $\boldsymbol{\alpha} = (\alpha_1, \alpha_2, \dots, \alpha_n) \in \mathbb{F}_{q^r}^n$, where no two elements among $\alpha_1, \alpha_2, \dots, \alpha_n$ form a conjugate pair over $\mathbb{F}_q$ and each $\alpha_i$ has exactly $r$ distinct conjugates over $\mathbb{F}_q$. Then for $1 \leq k \leq nr$, the code $\mathcal{ARS}_{k}(\boldsymbol{\alpha})$ is an additive MDS code of length $n$ and dimension $k$ over $\mathbb{F}_{q^r}$.
\end{lemma}

\begin{lemma}
\label{lemma3.2}
Suppose we have the following recurrence relation:

\begin{equation*}
\begin{cases}
c_{m-1} &= -f_{M-1},\\
c_{m-2} &= -f_{M-2} - f_{M-1}c_{m-1},\\
c_{m-3} &= -f_{M-3} - f_{M-2}c_{m-1} - f_{M-1}c_{m-2},\\
&\vdots\\
c_{0} &= -f_{M-m} - \displaystyle\sum_{j=1}^{m-1} f_{M-m+j}\,c_{m-j},
\end{cases}
\end{equation*}

i.e.

\begin{equation*}
\begin{cases}
c_{m-1} = -f_{M-1},\\
c_{m-j} = -f_{M-j} - \sum\limits_{l=1}^{j-1} f_{M-l} \cdot c_{m-(j-l)},\quad 2 \leq j \leq m.
\end{cases}
\end{equation*}

Then
\[
c_{m-j} = \sum_{\lambda=1}^{j} (-1)^{\lambda} \sum_{\substack{l_1+\cdots+l_{\lambda}=j \\ l_i\ge 1}} \prod_{i=1}^{\lambda} f_{M-l_i}, \quad 1 \leq j \leq m,
\tag{1}
\]
where the inner sum is over all ordered $\lambda$-tuples $(l_1, \ldots, l_\lambda)$ of positive integers summing to $j$.

\medskip
\begin{proof}
\noindent
We proceed with mathematical induction on the index $j$.

\begin{enumerate}
\item Base case: $j=1$.
The right-hand side of Equation (1) is
\[
 \sum_{\lambda=1}^{1} (-1)^{\lambda} \sum_{\substack{l_1+\cdots+l_{\lambda}=1 \\ l_i\ge 1}} \prod_{i=1}^{\lambda} f_{M-l_i} = -f_{M-1}=c_{m-1}.
\]

\item Inductive hypothesis: Suppose the identity (1) holds for every integer $s$ satisfying $1\le s<j$, that is
\[
c_{m-s} = \sum_{\lambda=1}^{s} (-1)^{\lambda} \sum_{\substack{l_1+\cdots+l_{\lambda}=s \\ l_i\ge 1}} \prod_{i=1}^{\lambda} f_{M-l_i},\quad 1\le s<j.
\]

\item Inductive step: We now establish the formula for the target case $j$. From the given recurrence relation and inductive hypothesis ($j>1$), we have
\[
c_{m-j} = -f_{M-j} - \sum_{l=1}^{j-1} f_{M-l} \cdot c_{m-(j-l)}, \quad 2 \leq j \leq m.
\]
Let $s = j-l$, which transforms the expression into
\[
c_{m-j} = -f_{M-j} - \sum_{s=1}^{j-1} f_{M-j+s} \cdot c_{m-s}.
\]
Substitute the inductive formula for $c_{m-s}$ into the above equality:
\[
c_{m-j} = -f_{M-j} - \sum_{s=1}^{j-1} f_{M-j+s} \left( \sum_{\lambda=1}^{s} (-1)^{\lambda} \sum_{\substack{l_1+\cdots+l_{\lambda}=s \\ l_i\ge 1}} \prod_{i=1}^{\lambda} f_{M-l_i} \right). \tag{2}
\]

We aim to show that the right-hand side of Equation (1) coincides with the right-hand side of Equation (2).
Denote the right-hand side of Equation (1) as $S_j$.
When $l_1=j$, we have $\lambda=1$, which yields
\[
(-1)^1 \sum_{l_1=j} f_{M-l_1} = -f_{M-j}.
\]
For the scenario $l_1 = j-s$ with $1\le s\le j-1$, we have $\lambda\ge 2$ and $l_2+\cdots+l_{\lambda}=s$.
Since $l_i \ge 1$ , we deduce the bound $\lambda-1 \le s$, or equivalently $\lambda \le s+1$.
Then we have
\[
\begin{aligned}
S_j &= \sum_{\lambda=1}^{j} (-1)^{\lambda} \sum_{\substack{l_1+\cdots+l_{\lambda}=j \\ l_i\ge 1}} \prod_{i=1}^{\lambda} f_{M-l_i} \\
&= -f_{M-j} + \sum_{\lambda=2}^{j} (-1)^{\lambda} \sum_{\substack{l_1+\cdots+l_{\lambda}=j \\ l_i\ge 1}} \prod_{i=1}^{\lambda} f_{M-l_i} \\
&= -f_{M-j} + \sum_{\lambda=2}^{j} (-1)^{\lambda} \sum_{s=\lambda-1}^{j-1} f_{M-(j-s)} \sum_{\substack{l_2+\cdots+l_{\lambda}=s \\ l_i\ge 1}} \prod_{i=2}^{\lambda} f_{M-l_i} \quad (since \quad \lambda-1 \le s) \\
&= -f_{M-j} + \sum_{s=1}^{j-1} \sum_{\lambda=2}^{s+1} (-1)^{\lambda} f_{M-(j-s)} \sum_{\substack{l_2+\cdots+l_{\lambda}=s \\ l_i\ge 1}} \prod_{i=2}^{\lambda} f_{M-l_i} \quad (\text{switch the order of summation})  \\
&= -f_{M-j} - \sum_{s=1}^{j-1} f_{M-(j-s)} \sum_{\lambda=2}^{s+1} (-1)^{\lambda-1} \sum_{\substack{l_2+\cdots+l_{\lambda}=s \\ l_i\ge 1}} \prod_{i=2}^{\lambda} f_{M-l_i} \\
&= -f_{M-j} - \sum_{s=1}^{j-1} f_{M-(j-s)} \sum_{\lambda'=1}^{s} (-1)^{\lambda'} \sum_{\substack{l_2+\cdots+l_{\lambda'+1}=s \\ l_i\ge 1}} \prod_{i=2}^{\lambda'+1} f_{M-l_i} \quad (let \quad \lambda-1=\lambda').
\end{aligned}
\]
\end{enumerate}
Since the indices are dummy variables, we may rename \( l_2, \ldots, l_{\lambda'+1} \) as \( l_1, \ldots, l_{\lambda'} \). Under this relabeling,
\[
\begin{aligned}
S_j&=-f_{M-j} - \sum_{s=1}^{j-1} f_{M-(j-s)} \sum_{\lambda'=1}^{s} (-1)^{\lambda'} \sum_{\substack{l_2+\cdots+l_{\lambda'+1}=s \\ l_i\ge 1}} \prod_{i=2}^{\lambda'+1} f_{M-l_i}
\\&=-f_{M-j} - \sum_{s=1}^{j-1} f_{M-(j-s)} \sum_{\lambda'=1}^{s} (-1)^{\lambda'} \sum_{\substack{l_1+\cdots+l_{\lambda'}=s \\ l_i\ge 1}} \prod_{i=1}^{\lambda'} f_{M-l_i}
\\&=-f_{M-j} - \sum_{s=1}^{j-1} f_{M-(j-s)} \sum_{\lambda=1}^{s} (-1)^{\lambda} \sum_{\substack{l_1+\cdots+l_{\lambda}=s \\ l_i\ge 1}} \prod_{i=1}^{\lambda} f_{M-l_i}
\\&=c_{m-j}.
\end{aligned}
\]
Thus
\[
c_{m-j} =S_j= \sum_{\lambda=1}^{j} (-1)^{\lambda} \sum_{\substack{l_1+\cdots+l_{\lambda}=j \\ l_i\ge 1}} \prod_{i=1}^{\lambda} f_{M-l_i}.
\]
This completes the induction.
\end{proof}
\end{lemma}

\noindent\textbf{Notation.}
Let $k=ar+b$ and $t=ur+v$ with $0\le b,v\le r-1$.
For $s$ in the range
\[
S=
\begin{cases}
\{0,1,\dots,u-1\}, & b=0,\ v=0,\\
\{0,1,\dots,u\}, & b=0,\ v\neq 0,\\
\{1,2,\dots,u\}, & b\neq 0,\ 1\le b+v\le r,\\
\{1,2,\dots,u+1\}, & b\neq 0,\ r+1\le b+v\le 2r-2,
\end{cases}
\]
put $M=(a+s)r=k+sr-b$ and $m=(k-1+t)-M=t-sr-1+b$.
Let $I\subseteq\{1,\dots,n\}$ with $|I|=a+s$, and set
\[
f_I(x)=\prod_{i\in I}\prod_{j=0}^{r-1}\bigl(x-\alpha_i^{q^j}\bigr)=\sum_{i=0}^{M}f_ix^i ,
\]
where $\boldsymbol{\alpha}=(\alpha_1,\alpha_2,\dots,\alpha_n)\in\mathbb{F}_{q^r}^n$ and $f_i$ is the coefficient of $x^i$ in $f_I(x)$, $0\le i\le M$ and $f_M=1$.
The coefficients $c_i$ of the polynomial
\[
c(x)=\sum_{i=0}^{m}c_ix^i\in\mathbb{F}_q[x]
\]
are fixed by $c_m=1$ together with the relations of Lemma \ref{lemma3.2}, namely
\[
c_{m-j}=\sum_{\lambda=1}^{j}(-1)^{\lambda}\!\!\sum_{\substack{l_1+\cdots+l_{\lambda}=j\\ l_i\ge1}}\ \prod_{i=1}^{\lambda}f_{M-l_i},\qquad 1\le j\le m,
\]
where the inner sum runs over all ordered $\lambda$-tuples $(l_1,\dots,l_\lambda)$ of positive integers summing to $j$; when $m=0$ this simply reads $c(x)=c_0=1$.
We emphasize that the coefficients $f_i$ and $c_i$ depend on both the subset $I$ and the evaluation vector $\boldsymbol{\alpha}$.
For brevity, we define the notation $P_s(I)$ as the logical condition given by
\[
P_s(I) :=
\begin{cases}
\text{vacuously true}, & s=0, \\[8pt]
\displaystyle\sum_{i_1+i_2=j} c_{i_1}f_{i_2} = 0,\ \forall\, k \le j \le M-1, & s \ne 0.
\end{cases}
\]

\begin{theorem}
\label{theorem3.2}
Let $q$ be a prime power and $n,k,t$ and $r \geq 2$ be positive integers such that $n \le |\mathfrak{B}|$, $1 \leq k < nr$, $1 \leq t \leq nr-k$, $h \in \{0,1,2,\dots,k-1\}$, $\eta \in \mathbb{F}_q^*$.
Let $k=ar+b$, $t=ur+v$, where $0 \leq b,v \leq r-1$, and let
\[
S=\begin{cases}
\{0,1,\dots,u-1\}, & b=0,\ v=0,\\[2pt]
\{0,1,\dots,u\}, & b=0,\ v\neq 0,\\[2pt]
\{1,2,\dots,u\}, & b\neq 0,\ 1\leq b+v\leq r,\\[2pt]
\{1,2,\dots,u+1\}, & b\neq 0,\ r+1\leq b+v\leq 2r-2.
\end{cases}
\]
Let $\boldsymbol{\alpha}=(\alpha_1,\alpha_2,\dots,\alpha_n)\in \mathbb{F}_{q^r}^n$, where no two elements among $\alpha_1,\alpha_2,\dots,\alpha_n$ form a conjugate pair over $\mathbb{F}_q$ and each $\alpha_i$ has exactly $r$ distinct conjugates over $\mathbb{F}_q$.
Then the code $\mathscr{T}_k(\boldsymbol{\alpha}, t, h, \eta)$ is an additive MDS code if and only if $(\eta,\boldsymbol{\alpha}) \notin \Lambda$, where
\[
\Lambda =\bigcup_{s\in S}\Lambda _s=\bigcup_{s\in S}\ \bigcup_{\substack{I \subseteq \{1, \dots, n\}\\ |I| = a+s}} \left\{(\eta, \boldsymbol{\alpha}) : P_s(I) \text{ holds and } \eta^{-1} = \sum\limits_{i_1+i_2=h}c_{i_1}f_{i_2}\in \mathbb{F}_q^*\right\},
\]
with $M=(a+s)r$ and $0 \leq i_1 \leq m$, $0 \leq i_2 \leq M$, where $m$ and $P_s(I)$ are as in the Notation preceding this theorem.
When $S=\varnothing$, which happens exactly when $b\neq 0$ and $1\le t\le r-b$, we have $\Lambda=\varnothing$ and the code $\mathscr{T}_k(\boldsymbol{\alpha},t,h,\eta)$ is an additive MDS code for every $\eta\in\mathbb{F}_q^*$.
\end{theorem}

\begin{proof}
\noindent
The code $\mathscr{T}_k(\boldsymbol{\alpha}, t, h, \eta)$ has dimension $k$ over $\mathbb{F}_q$.
By the Singleton bound for additive codes, we know that $d \leq n - \left\lceil \frac{k}{r} \right\rceil + 1$.
Since
\[
\operatorname{span}\{1, x,x^2, \dots,x^h+\eta x^{k-1+t},\dots, x^{k-1} \} \subseteq \operatorname{span}\{1, x, \dots, x^{k-1}, x^k, x^{k+1},\dots, x^{k-1+t}\},
\]
we have $\mathscr{T}_k(\boldsymbol{\alpha}, t, h, \eta) \subseteq \mathcal{ARS}_{k+t}(\boldsymbol{\alpha})$.
The code $\mathcal{ARS}_{k+t}(\boldsymbol{\alpha})$ is an additive MDS code by Lemma \ref{lemma3.1}, i.e. the minimum distance of $\mathcal{ARS}_{k+t}(\boldsymbol{\alpha})$ is $ n - \left\lceil \frac{k+t}{r} \right\rceil + 1$.
Thus
\[
n - \left\lceil \frac{k+t}{r} \right\rceil + 1 \leq d \leq n - \left\lceil \frac{k}{r} \right\rceil + 1.
\]
Then
\[
\left\lceil \frac{k}{r} \right\rceil=\left\lceil \frac{ar+b}{r} \right\rceil=a+\left\lceil \frac{b}{r} \right\rceil=
\begin{cases}
a, & b=0 \\
a+1, & 1 \leq b \leq r-1,
\end{cases}
\]
and, since $0 \leq b+v \leq 2r-2$,
\[
\left\lceil \frac{k+t}{r} \right\rceil=\left\lceil \frac{ar+b+ur+v}{r} \right\rceil=a+u+\left\lceil \frac{b+v}{r} \right\rceil=
\begin{cases}
a+u, & b+v=0 \\
a+u+1, & 1 \leq b+v \leq r \\
a+u+2, & r+1 \leq b+v \leq 2r-2,
\end{cases}
\]
where the third case is empty when $r=2$. Then we obtain
\[
\left\{
\begin{aligned}
n-(a+u-1) &\leq d \leq n-a+1, & b&=0,\, v = 0; \\
n-(a+u)   &\leq d \leq n-a+1, & b&=0,\, v \neq 0; \\
n-(a+u)   &\leq d \leq n-a,     & b&\neq 0,\, 1 \leq b+v \leq r; \\
n-(a+u+1) &\leq d \leq n-a,     & b&\neq 0,\, r+1 \leq b+v \leq 2r-2 .
\end{aligned}
\right.
\]
Comparing these intervals with the Singleton bound $n-\big\lceil k/r\big\rceil+1$, we see that in each of the four cases the code fails to be MDS if and only if $d=n-(a+s)$ for some $s\in S$; this is the only point at which the four cases are distinguished, and the argument below is uniform in $s\in S$.

In the degenerate case $S=\varnothing$, which occurs exactly when $b\neq 0$, $u=0$ and
$1\le b+v\le r$, no such $s$ exists, and the interval for $d$ in the third line reduces to the single value $d=n-a$, which coincides with the Singleton bound $n-\lceil k/r\rceil+1=n-(a+1)+1=n-a$. 
Hence $d$ attains the Singleton bound, the code is MDS, and correspondingly $\Lambda=\varnothing$: no value of $\eta$ has to be excluded. So there is nothing to prove in this case.

Fix $s\in S$ and suppose that the code $\mathscr{T}_k(\boldsymbol{\alpha}, t, h, \eta)$ is not an MDS code, so that $d=n-(a+s)$ for this $s$.
Choose a minimum-weight codeword with exactly $a+s$ zero components; it corresponds to a nonzero twisted polynomial
\[
g(x) = a_0 +a_1 x+ \cdots + a_{h} (x^{h}+\eta x^{k-1+t}) +\cdots + a_{k-1}x^{k-1}  \in \mathscr{P}_{n,k}(t,h,\eta).
\]
There exists a subset $I \subseteq \{1, 2, \dots, n\}$ with $|I| = a+s$ such that $g(\alpha_i) = 0$ for all $i \in I$.
Since \(g(x) \in \mathbb{F}_q[x]\), all conjugates of a root \(\alpha_i\) are also roots, yielding \((a + s)r\) pairwise distinct roots in total.
Thus, we have $g(\alpha_i^{q^j})=0$ for all $i \in I$ and $0 \leq j \leq r-1$, which means that $g(x)$ has at least $(a+s)r=k+sr-b$ roots.
Let $M=(a+s)r=k+sr-b$.

Two elementary inequalities, used below, hold for every $s\in S$.
\begin{itemize}
\item $M\ge k$, that is, $sr\ge b$. If $b=0$ this is immediate; if $b\neq 0$, then $s\ge 1$ by the definition of $S$, hence $sr-b\ge r-(r-1)>0$.
\item $m=(k-1+t)-M=(u-s)r+(b+v)-1\ge 0$. Indeed: if $b=0$, $v=0$, then $s\le u-1$ and $m\ge r-1$; if $b=0$, $v\neq 0$, then $s\le u$ and $m\ge v-1\ge 0$; if $b\neq 0$ and $1\le b+v\le r$, then $s\le u$ and $m\ge b+v-1\ge 0$; if $b\neq 0$ and $r+1\le b+v\le 2r-2$, then $s\le u+1$ and $m\ge b+v-r-1\ge 0$.
\end{itemize}
In particular $M\le k-1+t$.

Since $g(x)$ has degree at most $k-1+t$ and vanishes at the $M$ distinct points $\alpha_i^{q^j}$ with $M\ge k$, we conclude that $a_h \neq 0$. If $a_h=0$, then $\deg g(x) \leq k-1 < M$, which forces $g(x) \equiv 0$, contradicting that the codeword is nonzero. Hence $a_h \neq 0$ and $\deg g(x)=k-1+t$.

Let
\[
f_I(x) = \prod\limits_{i \in I} \prod\limits_{j=0}^{r-1} (x - \alpha_i^{q^j})=\sum\limits_{i=0}^{M}f_ix^i \in \mathbb{F}_q[x],
\]
where $f_{i}$ is the coefficient of $x^{i}$ in $f_I(x)$ and $f_{M}=1$.
Since $f_I$ is a product of pairwise distinct minimal polynomials over $\mathbb{F}_q$ and $g\in\mathbb{F}_q[x]$ vanishes at all its roots, $f_I$ divides $g(x)$ in $\mathbb{F}_q[x]$.
Let $m=(k-1+t)-M$, then $m=(k-1+t)-(k+sr-b)=t-sr-1+b=(u-s)r+b+v-1\ge0$, and we may write
$g(x)=\eta a_h c(x) f_I(x)$, where
\[
c(x)=\sum\limits_{i=0}^{t-sr-1+b}c_ix^i=\sum\limits_{i=0}^{m}c_ix^i\in \mathbb{F}_q[x],
\]
where $c_{i}$ is the coefficient of $x^{i}$ in $c(x)$ and $c_{t-sr-1+b}=c_m=1$.
Similarly, we write $g(x)=\sum\limits_{i=0}^{k-1+t}g_ix^i\in \mathbb{F}_q[x]$, where $g_{i}$ is the coefficient of $x^{i}$ in $g(x)$ and $g_{k-1+t}=\eta a_h$.

Since $g(x) \in \mathscr{P}_{n,k}(t,h,\eta)$, its coefficients satisfy
\[
g_{k-1+t} = \eta g_h,
\qquad\text{and}\qquad
g_j=0 \ \text{ for all } \ k\le j\le k+t-2,
\]
the second condition being empty when $t=1$.
By expanding $g(x)=\eta a_h c(x) f_I(x)$ and comparing the coefficients on the two sides, these conditions become the following equations:
\begin{equation}
\eta a_h=\eta\eta a_h\left(\sum\limits_{i_1+i_2=h}c_{i_1}f_{i_2}\right),
\label{eq1}
\end{equation}
\begin{equation}
\left\{
\begin{aligned}
g_{k+t-2}&=\eta a_h(c_{m-1}+f_{M-1})=0\\
g_{k+t-3}&=\eta a_h(c_{m-2}+c_{m-1}f_{M-1}+f_{M-2})=0\\
&\vdots\\
g_{M}&=\eta a_h(c_0+c_1f_{M-1}+\cdots+c_{m-1}f_{M-m+1}+f_{M-m})=0,
\end{aligned}
\right.
\label{eq2}
\end{equation}
\begin{equation}
\left\{
\begin{aligned}
g_{M-1}&=\eta a_h\left(\sum\limits_{i_1+i_2=M-1}c_{i_1}f_{i_2}\right)=0\\
&\vdots\\
g_{k}&=\eta a_h\left(\sum\limits_{i_1+i_2=k}c_{i_1}f_{i_2}\right)=0,
\end{aligned}
\right.
\label{eq3}
\end{equation}
where $c_{i}=0$ whenever $i<0$ or $i>m$ and $f_{i}=0$ whenever $i<0$ or $i>M$. Here \eqref{eq2} contains the $m$ equations indexed by $M\le j\le k+t-2$ and is empty when $m=0$, while \eqref{eq3} contains the $M-k=sr-b$ equations indexed by $k\le j\le M-1$ and is empty when $M=k$; together they account for all $k\le j\le k+t-2$.

Equation (\ref{eq2}) can be transformed into the following matrix form:
\begin{equation}
\begin{pmatrix}
1 & 0 & 0 & \cdots & 0 \\
f_{M-1} & 1 & 0 & \cdots & 0 \\
f_{M-2} & f_{M-1} & 1 & \cdots & 0 \\
\vdots & \vdots & \vdots & \ddots & \vdots \\
f_{M-m+1} & f_{M-m+2} & f_{M-m+3} & \cdots & 1
\end{pmatrix}
\begin{pmatrix}
c_{m-1} \\
c_{m-2} \\
c_{m-3} \\
\vdots \\
c_0
\end{pmatrix}
=
\begin{pmatrix}
-f_{M-1} \\
-f_{M-2} \\
-f_{M-3} \\
\vdots \\
-f_{M-m}
\end{pmatrix}
\label{eq4}
\end{equation}
As all diagonal entries of the lower triangular coefficient matrix equal $1$, its determinant is $1$, which guarantees a unique solution for $c_0, c_1, \dots, c_{m-1}$, which can be solved sequentially via forward substitution, i.e.
\begin{equation}
\begin{cases}
c_{m-1} &= -f_{M-1},\\
c_{m-2} &= -f_{M-2} - f_{M-1}c_{m-1},\\
c_{m-3} &= -f_{M-3} - f_{M-2}c_{m-1} - f_{M-1}c_{m-2},\\
&\vdots\\
c_{0} &= -f_{M-m} - \displaystyle\sum_{j=1}^{m-1} f_{M-m+j}\,c_{m-j}.
\end{cases}
\label{eq:c_solution_new}
\end{equation}
By Lemma \ref{lemma3.2}, then
\[
c_{m-j} = \sum_{\lambda=1}^{j} (-1)^{\lambda} \sum_{\substack{l_1+\cdots+l_{\lambda}=j \\ l_i\ge 1}} \prod_{i=1}^{\lambda} f_{M-l_i}, \quad 1 \leq j \leq m,
\]
where the inner sum is over all ordered $\lambda$-tuples $(l_1, \ldots, l_\lambda)$ of positive integers summing to $j$.
This reveals that \(c_i\) can be expressed as a polynomial in terms of \(f_i\).

Since
\[
\prod_{i\in I}\prod_{j=0}^{r-1} \big(x - \alpha_i^{q^j}\big)
= \sum_{i=0}^{M} f_i x^i
= x^{M} + \sum_{i=0}^{M-1} f_i x^i,
\]
and by Vieta's formulas, $f_{M-j}=(-1)^j e_j$, where $e_j$ denotes the $j$-th elementary symmetric sum evaluated over all $\alpha_i$ together with their $\mathbb{F}_q$-conjugates, $i \in I$, $|I| = a+s$ and $1 \leq j \leq M$.
In other words, once all $\alpha_i$ are given, all $f_i$ can be uniquely determined, which further yields all coefficients $c_i$.
Therefore, we assume in the following that all $c_i$ and $f_i$ are known quantities.

The conditions imposed by the remaining equations in (\ref{eq3}) depend only on the set $I$ and are denoted by $P_s(I)$; by definition $P_s(I)$ is the conjunction of the equations
\[
\sum_{i_1+i_2 = j} c_{i_1} f_{i_2} = 0, \quad k \le j \le M-1.
\]
When $M=k$, that is when $s=0$, this range is empty and $P_s(I)$ holds vacuously.

The equation (\ref{eq1}) becomes $\eta^{-1} =\sum\limits_{i_1+i_2=h}c_{i_1}f_{i_2} \in\mathbb{F}_q^*$.

Now, for $s\in S$, define
\[
\Lambda _s = \bigcup_{\substack{I \subseteq \{1, \dots, n\}\\ |I| = a+s}} \left\{(\eta, \boldsymbol{\alpha}) : P_s(I) \text{ holds and } \eta^{-1} = \sum\limits_{i_1+i_2=h}c_{i_1}f_{i_2}\right\}
\]
and
\[
\Lambda =\bigcup_{s\in S}\Lambda _s=\bigcup_{s\in S}\ \bigcup_{\substack{I \subseteq \{1, \dots, n\}\\ |I| = a+s}} \left\{(\eta, \boldsymbol{\alpha}) : P_s(I) \text{ holds and } \eta^{-1} = \sum\limits_{i_1+i_2=h}c_{i_1}f_{i_2}\right\}.
\]
So if the code $\mathscr{T}_k(\boldsymbol{\alpha},t,h,\eta)$ is not an MDS code, then $(\eta,\boldsymbol{\alpha}) \in \Lambda $.

Conversely, if $(\eta,\boldsymbol{\alpha}) \in \Lambda $, then there exist $s\in S$ and $I \subseteq \{1, 2, \dots, n\}$ with $|I| = a+s$ such that $P_s(I)$ holds and $\eta^{-1} = \sum\limits_{i_1+i_2=h}c_{i_1}f_{i_2}$.
Let $c_m=1$ and
\[
c_{m-j} = \sum_{\lambda=1}^{j} (-1)^{\lambda} \sum_{\substack{l_1+\cdots+l_{\lambda}=j \\ l_i\ge 1}} \prod_{i=1}^{\lambda} f_{M-l_i}, \quad 1 \leq j \leq m,
\]
then define $c(x)=\sum\limits_{i=0}^{m}c_ix^i\in \mathbb{F}_q[x]$.
Define $g(x)=\eta c(x) f_I(x)$, then a direct verification gives $g_{k-1+t}=\eta g_h$, $g_k= g_{k+1}=\cdots= g_{k+t-2}=0$ and $g_{k-1+t}=\eta$, which means that $0 \neq g(x) \in \mathscr{P}_{n,k}(t,h,\eta)$.
The twisted polynomial corresponds to a codeword $c_g=\left ( g(\alpha_1),g(\alpha_2),\dots,g(\alpha_n)\right )$.
Since $g(x)$ has $|I|= a+s$ zeros in $\{\alpha_1, \alpha_2, \dots, \alpha_n\}$, i.e. $c_g$ has at least $a+s$ zero components, we have
\[
d \leq \operatorname{wt}_H(c_g) \leq n-(a+s).
\]
If $b=0$, then $s\ge 0$ and $\left\lceil k/r\right\rceil=a$, hence $d\le n-a < n-a+1$.
If $b\neq 0$, then $s\ge 1$ and $\left\lceil k/r\right\rceil=a+1$, hence $d\le n-a-1 < n-a$.
In both cases $d$ is strictly smaller than the Singleton bound $n-\big\lceil k/r\big\rceil+1$, which implies that the code is not an MDS code.
\end{proof}

\begin{remark}
\label{remark3.1}
In \cite{1}, J. Ma et al. establish necessary and sufficient conditions for the code $\mathscr{T}_k(\boldsymbol{\alpha}, t, h, \eta)$ to be additive MDS code with $t=2$ and $h=k-1$.
In \cite{2},  Yadav and Sharma establish necessary and sufficient conditions for the code $\mathscr{T}_k(\boldsymbol{\alpha}, t, h, \eta)$ to be additive MDS code with $t=1$ and $0 \leq h \leq k-1$.
Our results include both scenarios.
\end{remark}

The following two examples illustrate the above theorem when $t=3$ and $t=4$.

\begin{example}
\label{example3.1}
Let $t=3$, $q = 16$, $r = 3$, $n = 8$, $k = 8$ and $h = 7$. Let $z$ be a primitive element of $\mathbb{F}_{16}$, and let $g$ be a primitive element of $\mathbb{F}_{4096}$.
Let $\boldsymbol{\alpha} = (g^{1163}, g^{1158}, g^{2029}, g^{1240}, g^{1639}, g^{119}, g^{2252}, g^{1531}) $ $\in \mathbb{F}_{4096}^8$.
A Magma computation gives that $\Lambda  =  \Lambda _1$,
where $\Lambda _1$ denotes
\[
\left\{
\begin{aligned}
f_7 &= f_8^2, \\
\eta^{-1} &\in \big\{ f_6 - f_7 \cdot f_8 \,\big\vert\, I \subseteq \{1,2,\dots,8\},\ |I|=3 \big\} \\
&= \big\{1, z^4, z^6, z^{10}, z^{12}\big\}.
\end{aligned}
\right.
\]

Then let $\eta = z^8 \in \mathbb{F}_{16}$ and $\eta^{-1}=z^7$, i.e. $(\eta,\boldsymbol{\alpha}) \notin \Lambda $.
From Theorem 3.2 we know that $\mathscr{T}_8(\boldsymbol{\alpha},3,7,z^8)$ is an additive $[8,8,6]$-code over $\mathbb{F}_{4096}$, and hence it is an additive MDS code.
\end{example}

\begin{example}
\label{example3.2}
Let $t=4$, $q = 16$, $r = 2$, $n = 6$, $k = 6$ and $h = 5$. Let $z$ be a primitive element of $\mathbb{F}_{16}$, and let $g$ be a primitive element of $\mathbb{F}_{256}$.
Let $\boldsymbol{\alpha} = (g^{157}, g^{125}, g^{90}, g^{140}, g^{58}, g^{143}) \in \mathbb{F}_{256}^6$. 
A Magma computation gives that  $\Lambda  = \Lambda _0 \cup \Lambda _1$,
where $\Lambda _0$ denotes
\[
\begin{aligned}
\eta^{-1} &\in \big\{ f_2 + f_4^2 + f_4 f_5^2 + f_5^4 \,\big\vert\, I \subseteq \{1,2,\dots,6\},\ |I|=3 \big\} \\
&= \big\{ z^2,\,z^4,\,z^7,\,z^9,\,z^{10},\,z^{11},\,z^{12},\,z^{13},\,z^{14} \big\}
\end{aligned}
\]
and $\Lambda _1$ denotes
\[
\left\{
\begin{aligned}
f_5 &= f_6 \cdot f_7,\\ f_6 &= f_7^2, \\
\eta^{-1} &\in \big\{ f_4 - f_5 f_7 \,\big\vert\, I \subseteq \{1,2,\dots,6\},\ |I|=4 \big\} \\
&= \big\{ z^4,\,z^7 \big\}.
\end{aligned}
\right.
\]
Then let $\eta = z^9 \in \mathbb{F}_{16}$ and $\eta^{-1}=z^6$, i.e. $(\eta,\boldsymbol{\alpha}) \notin \Lambda $.
From Theorem 3.2 we know that $\mathscr{T}_6(\boldsymbol{\alpha},4,5,z^9)$ is an additive $[6,6,4]$-code over $\mathbb{F}_{256}$, and hence it is an additive MDS code.
\end{example}

In general, the dual code of an additive GTRS code fails to be an additive GTRS code. Nevertheless, under certain conditions, we identify a class of additive GTRS codes, whose dual codes are also additive GTRS codes, where $h=k-1$.

\begin{lemma}
\label{lemma3.3}
Let $r_1,r_2,\dots,r_N$ be pairwise distinct elements. Define the polynomial
\[
P(X)=\prod_{t=1}^N (X-r_t).
\]
Then for every integer $0\le m \le N-2$, we have
\[
\sum_{t=1}^N \frac{r_t^m}{P'(r_t)} = 0.
\]

\begin{proof}
Consider the polynomial $f(X)=X^m$. The Lagrange interpolation polynomial of $f(X)$ at the nodes $r_1,r_2,\dots,r_N$ is given by
\[
L(X)=\sum_{t=1}^N \frac{r_t^m \cdot P(X)}{(X-r_t)P'(r_t)}.
\]
It is clear that $\deg L \le N-1$, and $L(r_k)=r_k^m$ for $1 \leq k \leq N$.
Set $Q(X)=X^m-L(X)$. Then $Q(r_k)=0$ for $1 \leq k \leq N$, which means that $Q(X)$ has $N$ distinct roots.
Since $\deg Q\le N-1$, $Q(X)\equiv 0$. Consequently,
\[
X^m = L(X)=\sum_{t=1}^N \frac{r_t^m \cdot P(X)}{(X-r_t)P'(r_t)}.
\]
Comparing the coefficients of $X^{N-1}$ on both sides: each term $\frac{P(X)}{X-r_t}$ is monic of degree $N-1$, so the right-hand side contributes $\sum\limits_{t=1}^N \frac{r_t^m}{P'(r_t)}$; the left-hand side has no $X^{N-1}$ term since $m\le N-2$, we have
\[
\sum_{t=1}^N \frac{r_t^m}{P'(r_t)} = 0,\quad 0\le m\le N-2.
\]
\end{proof}
\end{lemma}

\begin{lemma}
\label{lemma3.4}
Let $n \le |\mathfrak{B}|$, $\boldsymbol{\alpha}=(\alpha_1,\alpha_2,\dots,\alpha_n)\in\mathbb{F}_{q^r}^n$, where no two elements among $\alpha_1,\alpha_2,\dots,\alpha_n$ form a conjugate pair over $\mathbb{F}_q$, and each $\alpha_i$ has exactly $r$ distinct conjugates over $\mathbb{F}_q$. Let $N=nr$. Denote all $\alpha_i$ and their conjugates by $r_1,r_2,\dots,r_N$, and define the polynomial
\[
P(X)=\prod_{t=1}^N (X-r_t).
\]
Since the roots $r_t$ are pairwise distinct, we have $P'(\alpha_i) \neq 0$.
Let $\boldsymbol{u}=(u_1,u_2,\dots,u_n)\in\mathbb{F}_{q^r}^n$ with $u_i=\dfrac{1}{P'(\alpha_i)}$. 
Define
\[
S_m=\sum_{i=1}^n \operatorname{Tr}\big(u_i \alpha_i^m\big)=\sum_{i=1}^n \operatorname{Tr}\left(\frac{\alpha_i^m}{P'(\alpha_i)}\right).
\]
Then
\[
S_m=0,\quad 0\le m\le nr-2 \quad and \quad S_{nr-1}=1.
\]

\begin{proof}
By Lemma \ref{lemma3.3}, for any integer $0\le m\le N-2$, we have
\[
\sum_{t=1}^N \frac{r_t^m}{P'(r_t)}=0.
\]
For each $i$, let the conjugacy class of $\alpha_i$ be
\[
\mathrm{Conj}(\alpha_i)=\big\{\alpha_i,\alpha_i^q,\alpha_i^{q^2},\dots,\alpha_i^{q^{r-1}}\big\}.
\]
We partition the summation over all nodes by conjugacy classes, which yields
\[
\sum_{i=1}^n \sum_{r\in \mathrm{Conj}(\alpha_i)} \frac{r^m}{P'(r)}=0.
\]
Since the set of roots of $P(X)$ is closed under the Frobenius automorphism, we have $P(X)\in\mathbb{F}_q[x]$ and then $P'(X)\in\mathbb{F}_q[x]$.
For any polynomial $M(X)\in\mathbb{F}_q[X]$, the Frobenius map satisfies $\big(M(X)\big)^{q^s}=M(X^{q^s})$, then we have
\[
\big(P'(X)\big)^{q^s}=P'(X^{q^s}).
\]
Substituting $X=\alpha_i$, we obtain
\[
\big(P'(\alpha_i)\big)^{q^s}=P'\big(\alpha_i^{q^s}\big),\quad 0\le s\le r-1.
\]
Using this automorphism property, we expand the trace term:
\[
\begin{aligned}
\operatorname{Tr}\left(\frac{\alpha_i^m}{P'(\alpha_i)}\right)
&=\frac{\alpha_i^m}{P'(\alpha_i)}+\left(\frac{\alpha_i^m}{P'(\alpha_i)}\right)^q+\cdots+\left(\frac{\alpha_i^m}{P'(\alpha_i)}\right)^{q^{r-1}}\\
&=\frac{\alpha_i^m}{P'(\alpha_i)}+\frac{(\alpha_i^q)^m}{P'(\alpha_i^q)}+\cdots+\frac{(\alpha_i^{q^{r-1}})^m}{P'(\alpha_i^{q^{r-1}})}\\
&=\sum_{r\in \mathrm{Conj}(\alpha_i)} \frac{r^m}{P'(r)}.
\end{aligned}
\]
Therefore,
\[
S_m=\sum_{i=1}^n \operatorname{Tr}\left(\frac{\alpha_i^m}{P'(\alpha_i)}\right)
=\sum_{i=1}^n \sum_{r\in \mathrm{Conj}(\alpha_i)} \frac{r^m}{P'(r)}=0
\]
for all integers $0\le m\le N-2=nr-2$.

We derive this directly using the same Lagrange interpolation identity as in the proof of Lemma 3.3. Taking \(f(X) = X^{N-1}\), since \(\deg(f) = N-1 < N\), the Lagrange interpolation formula yields:
\[
X^{N-1} = \sum_{t=1}^N \frac{r_t^{N-1} P(X)}{(X - r_t) P'(r_t)},
\]
where \(P(X) = \prod_{t=1}^N (X - r_t)\). Since \(\frac{P(X)}{X - r_t}\) is a monic polynomial of degree \(N-1\), comparing the coefficients of \(X^{N-1}\) on both sides gives:
\[
1 = \sum_{t=1}^N \frac{r_t^{N-1}}{P'(r_t)}.
\]
By the definition of \(S_m\), the right-hand side is exactly \(S_{N-1}\). Hence, \(S_{nr-1} = 1\).
\qedhere
\end{proof}
\end{lemma}

We next present a construction of additive codes such that both the code and its dual are additive GTRS codes.
We consider the case when $\boldsymbol{v} = (1,1,\dots,1)$.

\begin{theorem}
\label{theorem3.3}
Let $q$ be a prime power and $n,k,t$ and $r \geq 2$ be positive integers such that $n \le |\mathfrak{B}|$, $1 \leq k \leq nr-1$, $1 \leq t \leq nr-k$, $h=k-1$, $\eta \in \mathbb{F}_q^*$.
Let $\boldsymbol{\alpha}=(\alpha_1,\alpha_2,\dots,\alpha_n)\in \mathbb{F}_{q^r}^n$, where no two elements among $\alpha_1,\alpha_2,\dots,\alpha_n$ form a conjugate pair over $\mathbb{F}_q$ and each $\alpha_i$ has exactly $r$ distinct conjugates over $\mathbb{F}_q$.
Let $\boldsymbol{u}=(u_1,u_2,\dots,u_n)\in(\mathbb{F}_{q^r}^*)^n$ with $u_i=\dfrac{1}{P'(\alpha_i)},    $where $P(x)=\prod\limits_{i=1}^n\prod\limits_{s=0}^{r-1}(x-\alpha_i^{q^s})$.
Define $K=nr-k$, $S_m=\sum\limits_{i=1}^n \operatorname{Tr}\left(\frac{\alpha_i^m}{P'(\alpha_i)}\right)$ for any integer $m\ge 0$ and $D=S_{nr-1}+\eta S_{nr-1+t}$.
\begin{enumerate}[(a)]
\item If $D \neq 0$,
\[
\mathscr{T}_k(\boldsymbol{\alpha},\boldsymbol{1},t,k-1,\eta)^\perp
= \mathscr{T}_{nr-k}(\boldsymbol{\alpha},\boldsymbol{u},\boldsymbol{t'},\boldsymbol{h'},\boldsymbol{\eta'}),
\]
where $\boldsymbol{1}$ is all-one vector, $ \boldsymbol{t'} = (\,\overbrace{1,\,1,\,\dots,\,1}^{t}\,)$, $ \boldsymbol{h'}=(K-t,K-t+1,\dots,K-1)$, $ \boldsymbol{\eta'}=(-\frac{\eta}{D} S_{nr-1},-\frac{\eta}{D} S_{nr},\cdots,-\frac{\eta}{D} S_{nr+t-2})$.

\item If $D=0$,\[
\mathscr{T}_k(\boldsymbol{\alpha},\boldsymbol{1},t,k-1,\eta)^\perp
= \big\{\big(u_1 f(\alpha_1),\dots,u_n f(\alpha_n)\big): f(x) \in \mathfrak{L}_{nr-k}\big\},
\]
where
\[
\mathfrak{L}_{nr-k}
= \left\{
\sum_{i=0}^{K-1-t} a_i x^i+\sum_{i=K-t}^{K-2} a_i \left(S_{i-K+t+nr }x^{K-t}-x^{i+1}\right)
+a_{K-1}x^{K}
\;:\; a_i\in\mathbb{F}_q
\right\} \subseteq \mathbb{F}_q[x].
\]
\end{enumerate}
\end{theorem}

\begin{proof}
\begin{enumerate}[(a)]
\item If $D \neq 0$, let $G$ be the generator matrix of $\mathscr{T}_k(\boldsymbol{\alpha},\boldsymbol{1},t,k-1,\eta)$
\[
G=
\begin{pmatrix}
1 & 1 & \cdots & 1 \\
\alpha_1 & \alpha_2 & \cdots & \alpha_n \\
\vdots & \vdots & \vdots & \vdots \\
\alpha_1^{k-2} & \alpha_2^{k-2} & \cdots & \alpha_n^{k-2} \\
\bigl(\alpha_1^{k-1}+\eta\alpha_1^{k-1+t}\bigr) & \bigl(\alpha_2^{k-1}+\eta\alpha_2^{k-1+t}\bigr) & \cdots & \bigl(\alpha_n^{k-1}+\eta\alpha_n^{k-1+t}\bigr)
\end{pmatrix}.
\]
We partition the rows of \(G\) into two groups:
\begin{itemize}
    \item \(G_1\) (ordinary rows): corresponding to \(x^i\) for \(0 \leq i \leq k-2\).
    \item \(G_2\) (special row): corresponding to \(x^{k-1} + \eta x^{k-1+t}\).
\end{itemize}
Then we consider the code \(\mathscr{T}_{nr-k}(\boldsymbol{\alpha}, \boldsymbol{u}, \boldsymbol{t}', \boldsymbol{h}', \boldsymbol{\eta}')\),  
where $ \boldsymbol{t'} = (\,\overbrace{1,\,1,\,\dots,\,1}^{t}\,)$, $ \boldsymbol{h'}=\{K-t,K-t+1,\dots,K-1\}$ and $ \boldsymbol{\eta'}=\{-\frac{\eta}{D} S_{nr-1},\cdots,-\frac{\eta}{D} S_{nr+t-2}\}$.
A generator matrix for this code is
\[
H=
\begin{pmatrix}
u_1 & u_2 & \cdots & u_n \\
u_1 \alpha_1 & u_2 \alpha_2 & \cdots & u_n \alpha_n \\
\vdots & \vdots & & \vdots \\
u_1 \alpha_1^{K-t-1} & u_2 \alpha_2^{K-t-1} & \cdots & u_n \alpha_n^{K-t-1} \\
u_1\left(\alpha_1^{K-t}-\dfrac{\eta}{D}\alpha_1^K\right) &
u_2\left(\alpha_2^{K-t}-\dfrac{\eta}{D}\alpha_2^K\right) &
\cdots &
u_n\left(\alpha_n^{K-t}-\dfrac{\eta}{D}\alpha_n^K\right) \\[6pt]
u_1\left(\alpha_1^{K-t+1}-\dfrac{\eta}{D}S_{nr}\alpha_1^K\right) &
u_2\left(\alpha_2^{K-t+1}-\dfrac{\eta}{D}S_{nr}\alpha_2^K\right) &
\cdots &
u_n\left(\alpha_n^{K-t+1}-\dfrac{\eta}{D}S_{nr}\alpha_n^K\right) \\
\vdots & \vdots & & \vdots \\
u_1\left(\alpha_1^{K-1}-\dfrac{\eta}{D}S_{nr+t-2}\alpha_1^K\right) &
u_2\left(\alpha_2^{K-1}-\dfrac{\eta}{D}S_{nr+t-2}\alpha_2^K\right) &
\cdots &
u_n\left(\alpha_n^{K-1}-\dfrac{\eta}{D}S_{nr+t-2}\alpha_n^K\right)
\end{pmatrix}
\]
where $D=S_{nr-1}+\eta S_{nr-1+t} \neq 0$ and $S_{nr-1}=1$ from Lemma \ref{lemma3.4}.
We partition its generator matrix \(H\) into two groups:
\begin{itemize}
    \item \(H_1\) (ordinary rows): corresponding to \(u_l x^j\) for \(0 \leq j \leq K-t-1\).
    \item \(H_2\) (special rows): corresponding to \(u_l (x^j + \eta'_j x^K)\) for \(K-t \leq j \leq K-1\). Here \(\eta'_j = -\frac{\eta}{D} S_{nr-1 + j - K + t}\), and \(D = S_{nr-1} + \eta S_{nr-1+t}\).
\end{itemize}

The $K$ rows of $H$ are linearly independent over $\mathbb{F}_q$. 
Indeed, after dividing the $l$-th column by $u_l\neq 0$, which does not affect the rank, the $j$-th row becomes the evaluation at $\boldsymbol{\alpha}$ of the polynomial $f_j=x^j$ for $0\le j\le K-t-1$ and of $f_j=x^j+\eta'_jx^K$ for $K-t\le j\le K-1$. 
Each monomial $x^j$ with $0\le j\le K-1$ occurs in $f_j$ with coefficient $1$ and in no other $f_{j'}$.
If $\sum_{j=0}^{K-1}c_jf_j$ vanished at $\alpha_1,\dots,\alpha_n$ with $c_j\in\mathbb{F}_q$, the resulting polynomial would have degree at most $K\le nr-1$ and would vanish at all $nr$ conjugates $\alpha_l^{q^s}$, hence would be zero; comparing the coefficients of $x^0,\dots,x^{K-1}$ then gives $c_0=\cdots=c_{K-1}=0$.

We verify the orthogonality ${\operatorname{Tr}}(GH\T) = 0$ by computing the trace inner products for all four block combinations:
\[
\operatorname{Tr}(G H\T) = 
\begin{pmatrix}
\operatorname{Tr}(G_1 H_1\T) & \operatorname{Tr}(G_1 H_2\T) \\
\operatorname{Tr}(G_2 H_1\T) & \operatorname{Tr}(G_2 H_2\T)
\end{pmatrix}.
\]

\textbf{Case 1: \(\operatorname{Tr}(G_1 H_1\T) = 0\).} 
The \((i, j)\)-th entry is \(\sum_{l=1}^n \operatorname{Tr}(\alpha_l^i \cdot u_l \alpha_l^j) = S_{i+j}\). 
Since \(0 \leq i \leq k-2\) and \(0 \leq j \leq K-t-1\), we have \(i+j \leq k-2 + nr-k-t-1 = nr-t-3 < nr-1\). 
By Lemma 3.4, \(S_m = 0\) for \(0 \leq m \leq nr-2\). Thus, all entries vanish.

\textbf{Case 2: \(\operatorname{Tr}(G_1 H_2\T) = 0\).} 
The \((i, j)\)-th entry is \(\sum_{l=1}^n \operatorname{Tr}(\alpha_l^i \cdot u_l (\alpha_l^j + \eta'_j \alpha_l^K)) = S_{i+j} + \eta'_j S_{i+K}\). 
Since \(i \leq k-2\), we have \(i+j \leq k-2 + K-1 = nr-3\) and \(i+K \leq k-2 + K = nr-2\). 
Thus all indices are at most \(nr-2\), and by Lemma 3.4, all entries are zero.

\textbf{Case 3: \(\operatorname{Tr}(G_2 H_1\T) = 0\).} 
The \(j\)-th entry is \(\sum_{l=1}^n \operatorname{Tr}((\alpha_l^{k-1} + \eta \alpha_l^{k-1+t}) \cdot u_l \alpha_l^j) = S_{k-1+j} + \eta S_{k-1+t+j}\). 
Since \(0 \leq j \leq K-t-1\), we have \(k-1+j \leq k-1 + nr-k-t-1 = nr-t-2 < nr-2\), and 
\(k-1+t+j \leq k-1+t + nr-k-t-1 = nr-2\). 
Thus, all indices are at most \(nr-2\), and by Lemma 3.4, all entries vanish.

\textbf{Case 4: \(\operatorname{Tr}(G_2 H_2\T) = 0\).} 
The \(j\)-th entry is 
\[
\sum_{l=1}^n \operatorname{Tr}((\alpha_l^{k-1} + \eta \alpha_l^{k-1+t}) \cdot u_l (\alpha_l^j + \eta'_j \alpha_l^K)).
\]
Expanding the product, we obtain four terms:
\[
S_{k-1+j} + \eta'_j S_{k-1+K} + \eta S_{k-1+t+j} + \eta \eta'_j S_{k-1+t+K}.
\]
Since \(k-1+K = nr-1\), \(k-1+t+K = nr-1+t\), and noting that \(S_m = 0\) for all \(m \leq nr-2\), the first term \(S_{k-1+j}\) vanishes because \(k-1+j \leq nr-2\). The remaining terms are:
\[
\eta'_j S_{nr-1} + \eta S_{k-1+t+j} + \eta \eta'_j S_{nr-1+t} = \eta'_j (S_{nr-1} + \eta S_{nr-1+t}) + \eta S_{k-1+t+j}.
\]
Substituting the definition of \(\eta'_j = -\frac{\eta}{D} S_{nr-1 + j - K + t}\) gives:
\[
\left(-\frac{\eta}{D} S_{nr-1 + j - K + t}\right) (S_{nr-1} + \eta S_{nr-1+t}) + \eta S_{nr-1 + j - K + t} = -\eta S_{nr-1 + j - K + t} + \eta S_{nr-1 + j - K + t} = 0.
\]
Thus, all entries of \(G_2 H_2\T\) are zero.

Combining the four cases, we have \(\operatorname{Tr}(GH\T) = 0\). Since the sum of their \(\mathbb{F}_q\)-dimension is \(k + (nr-k) = nr\), they are duals.

\item If $D=0$, let \(H\) be the generator matrix of this dual code, where
\[
H=
\begin{pmatrix}
u_1 & u_2 & \cdots & u_n \\
u_1\alpha_1 & u_2\alpha_2 & \cdots & u_n\alpha_n \\
\vdots & \vdots & & \vdots \\
u_1\alpha_1^{K-t-1} & u_2\alpha_2^{K-t-1} & \cdots & u_n\alpha_n^{K-t-1} \\
u_1\bigl(S_{nr}\alpha_1^{K-t}-\alpha_1^{K-t+1}\bigr) &
u_2\bigl(S_{nr}\alpha_2^{K-t}-\alpha_2^{K-t+1}\bigr) &
\cdots &
u_n\bigl(S_{nr}\alpha_n^{K-t}-\alpha_n^{K-t+1}\bigr) \\[6pt]
\vdots & \vdots & & \vdots \\
u_1\bigl(S_{nr+t-2}\alpha_1^{K-t}-\alpha_1^{K-1}\bigr) &
u_2\bigl(S_{nr+t-2}\alpha_2^{K-t}-\alpha_2^{K-1}\bigr) &
\cdots &
u_n\bigl(S_{nr+t-2}\alpha_n^{K-t}-\alpha_n^{K-1}\bigr) \\[6pt]
u_1\alpha_1^{K} & u_2\alpha_2^{K} & \cdots & u_n\alpha_n^{K}
\end{pmatrix}.
\]
We partition the rows of \(H\) into three groups:
\begin{itemize}
    \item \(H_1\) (ordinary rows): corresponding to \(u_l x^i\) for \(0 \leq i \leq K-t-1\).
    \item \(H_2\) (special rows): corresponding to \(u_l (S_{i-K+t+nr} x^{K-t} - x^{i+1})\) for \(K-t \leq i \leq K-2\).
    \item \(H_3\) (the last row): corresponding to \(u_l x^K\).
\end{itemize}

The $K$ rows of $H$ are linearly independent over $\mathbb{F}_q$. 
Indeed, after dividing the $l$-th column by $u_l\neq 0$, the rows are the evaluations at $\boldsymbol{\alpha}$ of the polynomials $x^i$ with $0\le i\le K-t-1$, of $S_{i-K+t+nr}x^{K-t}-x^{i+1}$ with $K-t\le i\le K-2$, and of $x^K$. 
The monomials $x^0,\dots,x^{K-t-1}$, $x^{K-t+1},\dots,x^{K-1}$ and $x^K$ are $K$ distinct monomials, and each of them occurs with a nonzero coefficient in exactly one of these polynomials. 
A linear relation among the rows would therefore produce a polynomial of degree at most $K\le nr-1$ vanishing at the $nr$ conjugates $\alpha_l^{q^s}$; this polynomial is zero, and comparing the coefficients of those $K$ monomials forces all coefficients of the relation to vanish.

Note that the \(\mathbb{F}_q\)-dimension of \(\mathfrak{L}_{nr-k}\) is exactly \((K-t) + (t-1) + 1 = K = nr-k\). 
We verify the orthogonality \(\operatorname{Tr}(GH\T) = 0\) by computing the trace inner products for the block combinations of \(G\) and \(H\):

\textbf{Case 1: \(\operatorname{Tr}(G_1 H_1\T) = 0\) and \(\operatorname{Tr}(G_1 H_3\T) = 0\).} 
For \(G_1\) (corresponding to \(x^j\) for \(0 \leq j \leq k-2\)), the entries are \(S_{j+i}\) for \(H_1\) and \(S_{j+K}\) for \(H_3\). Since \(i \leq K-t-1\) and \(j \leq k-2\), both \(j+i \leq nr-t-3\) and \(j+K \leq nr-2\). By Lemma 3.4, all terms vanish.

\textbf{Case 2: \(\operatorname{Tr}(G_1 H_2\T) = 0\).} 
The \((j, i)\)-th entry is \(\operatorname{Tr}(x^j \cdot u_l(S_{i-K+t+nr} x^{K-t} - x^{i+1})) = S_{i-K+t+nr} S_{j+K-t} - S_{j+i+1}\). Since \(j \leq k-2\) and \(i \leq K-2\), we have \(j+K-t \leq nr-t-2\) and \(j+i+1 \leq nr-3\). By Lemma 3.4, all entries are zero.

\textbf{Case 3: \(\operatorname{Tr}(G_2 H_1\T) = 0\).} 
For \(G_2\) (corresponding to \(x^{k-1} + \eta x^{k-1+t}\)), the \(i\)-th entry is \(S_{k-1+i} + \eta S_{k-1+t+i}\). Since \(i \leq K-t-1\), the maximum indices are \(nr-t-2\) and \(nr-2\), so all entries vanish.

\textbf{Case 4: \(\operatorname{Tr}(G_2 H_2\T) = 0\).} 
The \(i\)-th entry (\(K-t \leq i \leq K-2\)) is:
\[
\operatorname{Tr}\big((x^{k-1} + \eta x^{k-1+t}) \cdot u_l(S_{i-K+t+nr} x^{K-t} - x^{i+1})\big).
\]
Expanding the product, we obtain:
\[
S_{i-K+t+nr} S_{k-1+K-t} + \eta S_{i-K+t+nr} S_{k-1+t+K-t} - S_{k+i} - \eta S_{k+t+i}.
\]
Since \(k-1+K-t = nr-t-1\) and \(k-1+t+K-t = nr-1\), and noting that \(t \ge 1 \implies S_{nr-t-1} = 0\), \(S_{nr-1} = 1\), the first two terms reduce to \(\eta S_{i-K+t+nr}\). 
Observe that the index \(i-K+t+nr = i-(nr-k)+t+nr = k+t+i\). Substituting this yields:
\[
\eta S_{k+t+i} - S_{k+i} - \eta S_{k+t+i} = -S_{k+i}.
\]
Since \(i \leq K-2 = nr-k-2\), we have \(k+i \leq nr-2\), so \(S_{k+i} = 0\). Thus, this block vanishes.

\textbf{Case 5: \(\operatorname{Tr}(G_2 H_3\T) = 0\).} 
The entry is \(\operatorname{Tr}\big((x^{k-1} + \eta x^{k-1+t}) \cdot u_l x^K\big) = S_{k-1+K} + \eta S_{k-1+t+K}\). 
Since \(k-1+K = nr-1\) and \(k-1+t+K = nr-1+t\), this is exactly \(S_{nr-1} + \eta S_{nr-1+t} = D\). 
Because we are in the case \(D=0\), this entry is zero.

Combining all cases, we have \(\operatorname{Tr}(GH\T) = 0\). Since the sum of their \(\mathbb{F}_q\)-dimensions is \(k + (nr-k) = nr\), they are duals.
\end{enumerate}
\end{proof}

\begin{remark}
\label{remark3.2}
Note that \(\eta'_j \) may be zero. If \(\eta'_j = 0\), the corresponding twist term vanishes, and the dual code degenerates to a GTRS code with fewer twists, where the \(j\)-th components of the twist vector \(\boldsymbol{t}'\), hook vector \(\boldsymbol{h}'\), and coefficient vector \(\boldsymbol{\eta}'\) are synchronously deleted. This does not affect the orthogonality.
\end{remark}

The following  example illustrates the above theorem.

\begin{example}
\label{example3.3}
Let $t=4$, $q = 16$, $r = 2$, $n = 8$, $k = 8$, $h = 7$, then $K=8$.
Let $z$ be a primitive element of $\mathbb{F}_{16}$, and let $g$ be a primitive element of $\mathbb{F}_{256}$.
Let $\eta=z^5$ and $\boldsymbol{\alpha} = (g, g^{18}, g^{35}, g^{52}, g^{69}, g^{86}, g^{103}, g^{120}) \in \mathbb{F}_{256}^8$.
From Theorem 3.3 we have $D=z \ne 0$,
$ \boldsymbol{t'} = (1,1,1,1)$, $ \boldsymbol{h'}=\{4,5,6,7\}$, $ \boldsymbol{\eta'}=\{z^4,z^4,z^2,z^{11}\}$ and $\boldsymbol{u}=(g^{112}, g^{88}, g, g^{54}, g^{85}, g^{19}, g^{153}, g^{44})$.
Using Magma, we have $\mathscr{T}_8(\boldsymbol{\alpha},\boldsymbol{1},4,7,z^5)^\perp
= \mathscr{T}_{8}(\boldsymbol{\alpha},\boldsymbol{u},\boldsymbol{t'},\boldsymbol{h'},\boldsymbol{\eta'}),
$
where $\boldsymbol{1}$ is all-one vector.
Both codes are additive GTRS codes.
\end{example}

\section{ACD codes from additive GTRS codes over finite fields}
In this section we present a construction of ACD codes from GTRS codes.

\begin{theorem}
\label{theorem4.1}
Let $q$ be a prime power and $n,k,t$ and $r \geq 2$ be positive integers such that $n \le |\mathfrak{B}|$, $3 \leq k \leq nr-1$, $1 \leq t \leq nr-k$, $h=k-1$, $\eta \in \mathbb{F}_q^*$.
Let $\boldsymbol{\alpha}=(\alpha_1,\alpha_2,\dots,\alpha_n)\in \mathbb{F}_{q^r}^n$, where no two elements among $\alpha_1,\alpha_2,\dots,\alpha_n$ form a conjugate pair over $\mathbb{F}_q$ and each $\alpha_i$ has exactly $r$ distinct conjugates over $\mathbb{F}_q$.
Denote $T_m=\sum\limits_{i=1}^n \operatorname{Tr}(\alpha_i^m)\in \mathbb{F}_q$, $A=
\big(T_{i+j}\big)_{0\le i,j\le k-2}$ and suppose $\det(A) \neq 0$, $\boldsymbol{w} = \big(\, T_{k-1}+\eta T_{k-1+t},\; T_{k}+\eta T_{k+t},\;\dots,\; T_{2k-3}+\eta T_{2k-3+t} \,\big)\T$ and $\Delta=T_{2k-2}+2\eta T_{2k-2+t}+\eta^2 T_{2k-2+2t}-\boldsymbol{w}\T A^{-1}\boldsymbol{w}$.
Then $\mathscr{T}_k(\boldsymbol{\alpha},1,t,k-1,\eta)$ is an ACD code if and only if $\Delta \neq 0$.

\begin{proof}
Let $G$ be the generator matrix of $\mathscr{T}_k(\boldsymbol{\alpha},\boldsymbol{1},t,k-1,\eta)$
\[
G=
\begin{pmatrix}
1 & 1 & \cdots & 1 \\
\alpha_1 & \alpha_2 & \cdots & \alpha_n \\
\vdots & \vdots & \vdots & \vdots \\
\alpha_1^{k-2} & \alpha_2^{k-2} & \cdots & \alpha_n^{k-2} \\
\bigl(\alpha_1^{k-1}+\eta\alpha_1^{k-1+t}\bigr) & \bigl(\alpha_2^{k-1}+\eta\alpha_2^{k-1+t}\bigr) & \cdots & \bigl(\alpha_n^{k-1}+\eta\alpha_n^{k-1+t}\bigr)
\end{pmatrix}.
\]
Then
\[
\operatorname{Tr}\bigl(GG\T\bigr)=
\begin{pmatrix}
T_0 & T_1 & \cdots & T_{k-2} & T_{k-1}+\eta T_{k-1+t} \\
T_1 & T_2 & \cdots & T_{k-1} & T_{k}+\eta T_{k+t} \\
\vdots & \vdots & & \vdots & \vdots \\
T_{k-2} & T_{k-1} & \cdots & T_{2k-4} & T_{2k-3}+\eta T_{2k-3+t} \\
T_{k-1}+\eta T_{k-1+t} & T_{k}+\eta T_{k+t} & \cdots & T_{2k-3}+\eta T_{2k-3+t} & T_{2k-2}+2\eta T_{2k-2+t}+\eta^2 T_{2k-2+2t}
\end{pmatrix}.
\]
Let
\[
\operatorname{Tr}\bigl(GG\T\bigr)
=
\begin{pmatrix}
\boldsymbol{A} & \boldsymbol{w} \\
\boldsymbol{w}\T & d
\end{pmatrix},
\]
where
\[
\boldsymbol{A}=
\begin{pmatrix}
T_0 & T_1 & \cdots & T_{k-2} \\
T_1 & T_2 & \cdots & T_{k-1} \\
\vdots & \vdots & \ddots & \vdots \\
T_{k-2} & T_{k-1} & \cdots & T_{2k-4}
\end{pmatrix}= \big(T_{i+j}\big)_{0\le i,j\le k-2}
,
\boldsymbol{w}=
\begin{pmatrix}
T_{k-1}+\eta T_{k-1+t} \\
T_{k}+\eta T_{k+t} \\
\vdots \\
T_{2k-3}+\eta T_{2k-3+t}
\end{pmatrix},
\]
and $d = T_{2k-2}+2\eta T_{2k-2+t}+\eta^2 T_{2k-2+2t}$.
Let
\[
M = \operatorname{Tr}\bigl(GG\T\bigr)
=
\begin{pmatrix}
A & \boldsymbol{w} \\
\boldsymbol{w}\T & d
\end{pmatrix}.
\]
Since the submatrix $A$ is invertible, the inverse matrix $A^{-1}$ exists. We perform block elementary row elimination by multiplying a block lower triangular matrix on the left:
\[
\begin{pmatrix}
\boldsymbol{I}_{k-1} & \boldsymbol{0} \\
-\boldsymbol{w}\T A^{-1} & 1
\end{pmatrix}
\begin{pmatrix}
A & \boldsymbol{w} \\
\boldsymbol{w}\T & d
\end{pmatrix}
=
\begin{pmatrix}
A & \boldsymbol{w} \\
\boldsymbol{0} & d - \boldsymbol{w}\T A^{-1}\boldsymbol{w}
\end{pmatrix},
\]
where $\boldsymbol{I}_{k-1}$ denotes the $(k-1)$-order identity matrix.

Then
\[
\det(M) = \det(A)\cdot \bigl(d - \boldsymbol{w}\T A^{-1}\boldsymbol{w}\bigr).
\]
Let $\Delta=d - \boldsymbol{w}\T A^{-1}\boldsymbol{w}=T_{2k-2}+2\eta T_{2k-2+t}+\eta^2 T_{2k-2+2t}-\boldsymbol{w}\T A^{-1}\boldsymbol{w}$.
If $\det(A)\neq 0$ and the Schur complement $\Delta \neq 0$, then $\det(M)\neq 0$.
Combined with Lemma \ref{lemma2.2}, the code $\mathscr{T}_k(\boldsymbol{\alpha},\boldsymbol{1},t,k-1,\eta)$ is an ACD code.

Conversely, if $\mathscr{T}_k(\boldsymbol{\alpha},1,t,k-1,\eta)$ is an ACD code, then $\det(M) \neq 0$. Since $\det(A) \neq 0$, we immediately have
$\Delta \neq 0$.
\end{proof}

\end{theorem}

Next, we prove the existence of the evaluation points \(\boldsymbol{\alpha}\) such that the matrix \(A\) in Theorem 4.1 is invertible and \(\Delta \neq 0\).

\begin{remark}
\label{remark4.1}
Let $T_m=\sum\limits_{i=1}^n \operatorname{Tr}(\alpha_i^m)$ and     $A=\bigl(T_{i+j}\bigr)_{0\le i,j\le k-2}$, where $\boldsymbol{\alpha}=(\alpha_1,\alpha_2,\dots,\alpha_n)\in \mathbb{F}_{q^r}^n$.
Let $\boldsymbol{w}_0 = \big( T_{k-1},T_{k},\dots, T_{2k-3} \,\big)\T$ and $C_0=\bigl(T_{2k-2}-\boldsymbol{w}_0\T A^{-1}\boldsymbol{w}_0\bigr)$ when the matrix $A$ is invertible. 
Let $B_0 = (T_{i+j})_{0 \le i,j \le k-1}$
be the \(k \times k\) Hankel matrix obtained from \(A\) by adjoining the row and the column indexed by \(k-1\). 
By applying the Schur complement formula to the block partition of $B_0$: 
\[
\det(B_0)=\det(A)\bigl(T_{2k-2}-\boldsymbol{w}_0\T A^{-1}\boldsymbol{w}_0\bigr)=\det(A)\,C_0.
\]
\end{remark}
\begin{lemma}
\label{lemma4.1}
Fix positive integers $n,r,k$ with $r\ge 2$ and $k\le nr$, and let $N=nr$.
Then for every sufficiently large prime power \(q\) (allowing the exclusion of finitely many characteristics), there exist $\alpha_1,\dots,\alpha_n\in\mathbb{F}_{q^r}$ such that no two elements among $\alpha_1,\alpha_2,\dots,\alpha_n$ form a conjugate pair over $\mathbb{F}_q$, each $\alpha_i$ has exactly $r$ distinct conjugates over $\mathbb{F}_q$, and
\[
\det(A)\neq 0
\qquad\text{and}\qquad
\det(B_0)\neq 0,
\]
where $A=\bigl(T_{i+j}\bigr)_{0\le i,j\le k-2}$, $B_0=\bigl(T_{i+j}\bigr)_{0\le i,j\le k-1}$ and $T_m=\sum\limits_{i=1}^n \operatorname{Tr}(\alpha_i^m)$.
\end{lemma}

\begin{proof}
Let $s=k-1$. Assume that the conjugacy orbits of $\alpha_1,\dots,\alpha_n$ are pairwise disjoint and denote
\[
\{\beta_1,\dots,\beta_N\}=\bigcup_{i=1}^n\{\alpha_i,\alpha_i^q,\dots,\alpha_i^{q^{r-1}}\}.
\]
Then
\[
T_m
= \sum_{i=1}^n \operatorname{Tr}(\alpha_i^m)
= \sum_{i=1}^n \sum_{j=0}^{r-1} \bigl(\alpha_i^{q^j}\bigr)^m
= \sum_{\ell=1}^N \beta_\ell^m,
\]
and therefore
\[
A=\Bigl(\sum_{\ell=1}^N \beta_\ell^{i+j}\Bigr)_{0\le i,j\le s-1},
\qquad
B_0=\Bigl(\sum_{\ell=1}^N \beta_\ell^{i+j}\Bigr)_{0\le i,j\le s}.
\]
Define the $s\times N$ matrix $V=\bigl(\beta_\ell^{i}\bigr)_{0\le i\le s-1,\,1\le\ell\le N}$ and the $(s+1)\times N$ matrix $U=\bigl(\beta_\ell^{i}\bigr)_{0\le i\le s,\,1\le\ell\le N}$. By direct computation we have
\[
A=VV\T,\qquad B_0=UU\T .
\]
Note that over a finite field one cannot deduce that $VV\T$ or $UU\T$ is invertible merely from the fact that $V$ and $U$ have full row rank. We therefore employ the polynomial counting method.

Temporarily treat $\beta_1,\dots,\beta_N$ as independent indeterminates, and define
\[
D_A(\beta_1,\dots,\beta_N)=\det\Bigl(\sum_{\ell=1}^N \beta_\ell^{i+j}\Bigr)_{0\le i,j\le s-1},
\qquad
D_B(\beta_1,\dots,\beta_N)=\det\Bigl(\sum_{\ell=1}^N \beta_\ell^{i+j}\Bigr)_{0\le i,j\le s}.
\]
We first prove that $D_A$ and $D_B$ are nonzero polynomials with integer coefficients.

In the real number field, choose mutually distinct positive real numbers $0<\beta_1<\beta_2<\cdots<\beta_N$. Since $N\ge s+1$, the first $s+1$ columns of $U$ form an $(s+1)\times (s+1)$ Vandermonde matrix whose determinant is $\prod\limits_{1\le i<j\le s+1}(\beta_j-\beta_i)\neq 0$; hence $\operatorname{rank}_{\mathbb{R}}U=s+1$. For any nonzero column vector $x\in\mathbb{R}^{s+1}$ we then have $x\T UU\T x=\|U\T x\|^2>0$, so $UU\T$ is positive definite over $\mathbb{R}$ and
\[
D_B(\beta_1,\dots,\beta_N)=\det(UU\T)>0 .
\]
The same argument applied to $V$ gives $D_A(\beta_1,\dots,\beta_N)=\det(VV\T)>0$. This shows that $D_A$ and $D_B$ are nonzero polynomials with integer coefficients.

Since each power sum $\sum\limits_{\ell=1}^N \beta_\ell^m$ is a symmetric polynomial in $\beta_1,\dots,\beta_N$, both $D_A$ and $D_B$ are symmetric polynomials. Consider the monic polynomial
\[
F(X)=\prod_{\ell=1}^N (X-\beta_\ell)=X^N+c_1X^{N-1}+\cdots+c_N .
\]
By the fundamental theorem of symmetric polynomials, there exist nonzero polynomials
\[
P_A(c_1,\dots,c_N),\ P_B(c_1,\dots,c_N)\in\mathbb{Z}[c_1,\dots,c_N]
\]
such that
\[
D_A(\beta_1,\dots,\beta_N)=P_A(c_1,\dots,c_N),\qquad
D_B(\beta_1,\dots,\beta_N)=P_B(c_1,\dots,c_N).
\]
Let $d_A$ and $d_B$ be their total degrees; both depend only on $N=nr$ and $s=k-1$, hence only on $n,r,k$, but not on $q$. 
Since $P_AP_B$ is a nonzero polynomial with integer coefficients, only finitely many primes $p$ satisfy $P_AP_B\equiv 0\pmod p$. We fix a characteristic $p$ outside this finite set and from now on take $q=p^f$ for some integer $f\ge 1$.

Let $I_q(r)$ denote the number of monic irreducible polynomials of degree $r$ in $\mathbb{F}_q[X]$. It is known that
\[
I_q(r)=\frac{1}{r}\sum_{e\mid r}\mu(e)q^{r/e}=\frac{q^r}{r}+O(q^{r/2})
\]
as $q\to\infty$ with $r$ fixed. Hence, for all sufficiently large prime powers q of characteristics outside a finite set, $I_q(r)\ge n$, and the number of monic polynomials of degree $N$ which are products of $n$ distinct monic irreducible polynomials of degree $r$ satisfies
\[
\binom{I_q(r)}{n}=\frac{q^{N}}{n!\,r^n}+O(q^{N-1})\ \ge\ c\,q^{N}
\]
for some constant $c>0$ independent of $q$, where one may take $c=\frac{1}{2n!\,r^n}$ for $q$ large enough; in particular $c$ depends only on $n$ and $r$, not on $q$. Likewise, the number $I_q(r)$ exceeds $n$ as soon as $q$ is larger than a bound depending only on $n$ and $r$. 
Every such $F$ is squarefree and has exactly $n$ distinct irreducible factors of degree $r$. Moreover, the coefficient vector $(c_1,\dots,c_N)$ determines $F$ uniquely, so distinct such polynomials give distinct points of $\mathbb{F}_q^N$.

Since $P_AP_B\not\equiv 0$ over $\mathbb{F}_q$ and $\deg(P_AP_B)=d_A+d_B$, the Schwartz--Zippel root estimate over finite fields yields
\[
\bigl|\{(c_1,\dots,c_N)\in\mathbb{F}_q^N : P_A(c_1,\dots,c_N)P_B(c_1,\dots,c_N)=0\}\bigr|
\le (d_A+d_B)\,q^{N-1}.
\]
For $q$ sufficiently large such that $c\,q^{N}>(d_A+d_B)q^{N-1}$, that is $q>\frac{d_A+d_B}{c}$, not all admissible coefficient vectors lie in the zero set of $P_AP_B$. 
Therefore, there exists an admissible polynomial $F(X)=f_1(X)\cdots f_n(X)$ with
\[
P_A(c_1,\dots,c_N)\neq 0
\qquad\text{and}\qquad
P_B(c_1,\dots,c_N)\neq 0 .
\]

For each $i$, choose a root $\alpha_i$ of $f_i(X)$ in $\mathbb{F}_{q^r}$. Since $f_i$ is irreducible over $\mathbb{F}_q$ of degree $r$, we have $[\mathbb{F}_q(\alpha_i):\mathbb{F}_q]=r$, so $\alpha_i$ has exactly $r$ distinct conjugates $\alpha_i,\alpha_i^q,\dots,\alpha_i^{q^{r-1}}$. Since $f_1,\dots,f_n$ are pairwise distinct, the Frobenius orbits of distinct $\alpha_i$ are disjoint; in other words, for $i\neq j$ the elements $\alpha_i$ and $\alpha_j$ are not conjugates over $\mathbb{F}_q$. Finally, $\{\beta_1,\dots,\beta_N\}$ is precisely the union of these orbits, and hence
\[
\det(A)=D_A(\beta_1,\dots,\beta_N)=P_A(c_1,\dots,c_N)\neq 0,
\det(B_0)=D_B(\beta_1,\dots,\beta_N)=P_B(c_1,\dots,c_N)\neq 0 .
\]
This completes the proof.
\end{proof}

\begin{corollary}\label{corollary4.1}
Fix positive integers \(n, r, k, t\) with \(r \ge 2\), \(3 \le k \le nr - 1\), and \(1 \le t \le nr - k\). For every sufficiently large prime power \(q\) ($q \geq 4$) whose characteristic avoids a finite set, there exist \(\alpha_1, \dots, \alpha_n \in \mathbb{F}_{q^r}\) lying in pairwise distinct full-degree Frobenius orbits and an \(\eta \in \mathbb{F}_q^*\) such that \(\mathscr{T}_k(\boldsymbol{\alpha}, \mathbf{1}, t, k - 1, \eta)\) is an ACD code.
\end{corollary}

\begin{proof}
By Lemma \ref{lemma4.1}, there exist \(\alpha_1, \dots, \alpha_n \in \mathbb{F}_{q^r}\) lying in pairwise distinct full-degree Frobenius orbits s.t. $\det(A)\neq 0$, $\det(B_0)\neq 0$.
Since $\det(B_0)=\det(A)\,C_0$, $C_0$ is nonzero.
We expand $\Delta$ in powers of $\eta$.
Recall that
\[
\Delta = d - \boldsymbol{w}\T A^{-1}\boldsymbol{w},
\]
where
\[
d = T_{2k-2} + 2\eta T_{2k-2+t} + \eta^2 T_{2k-2+2t},\quad
\boldsymbol{w} = \boldsymbol{w}_0 + \eta \boldsymbol{w}_1,
\]
with constant vectors independent of $\eta$
\[
\boldsymbol{w}_0 = \begin{pmatrix} T_{k-1} \\ T_{k} \\ \vdots \\ T_{2k-3} \end{pmatrix},\quad
\boldsymbol{w}_1 = \begin{pmatrix} T_{k-1+t} \\ T_{k+t} \\ \vdots \\ T_{2k-3+t} \end{pmatrix}.
\]
Then
\[
\begin{aligned}
\boldsymbol{w}\T A^{-1}\boldsymbol{w}
&= (\boldsymbol{w}_0 + \eta \boldsymbol{w}_1)\T A^{-1}(\boldsymbol{w}_0 + \eta \boldsymbol{w}_1) \\
&= (\boldsymbol{w}_0\T + \eta \boldsymbol{w}_1\T)A^{-1}(\boldsymbol{w}_0 + \eta \boldsymbol{w}_1) \\
&= \boldsymbol{w}_0\T A^{-1}\boldsymbol{w}_0
+ \boldsymbol{w}_0\T A^{-1}(\eta \boldsymbol{w}_1)
+ \eta \boldsymbol{w}_1\T A^{-1}\boldsymbol{w}_0
+ \eta \boldsymbol{w}_1\T A^{-1}(\eta \boldsymbol{w}_1) \\
&= \boldsymbol{w}_0\T A^{-1}\boldsymbol{w}_0
+ \eta\big(\boldsymbol{w}_0\T A^{-1}\boldsymbol{w}_1 + \boldsymbol{w}_1\T A^{-1}\boldsymbol{w}_0\big)
+ \eta^2 \boldsymbol{w}_1\T A^{-1}\boldsymbol{w}_1.
\end{aligned}
\]
We have
\[
\begin{aligned}
\Delta(\eta)
&= \big(T_{2k-2} + 2T_{2k-2+t}\eta + T_{2k-2+2t}\eta^2\big) \\
&\quad - \Big[\boldsymbol{w}_0\T A^{-1}\boldsymbol{w}_0
+ \eta\big(\boldsymbol{w}_0\T A^{-1}\boldsymbol{w}_1 + \boldsymbol{w}_1\T A^{-1}\boldsymbol{w}_0\big)
+ \eta^2 \boldsymbol{w}_1\T A^{-1}\boldsymbol{w}_1\Big] \\
&= \underbrace{\Big(T_{2k-2} - \boldsymbol{w}_0\T A^{-1}\boldsymbol{w}_0\Big)}_{C_0}
+ \underbrace{\Big(2T_{2k-2+t} - \boldsymbol{w}_0\T A^{-1}\boldsymbol{w}_1 - \boldsymbol{w}_1\T A^{-1}\boldsymbol{w}_0\Big)}_{C_1}\,\eta
+ \underbrace{\Big(T_{2k-2+2t} - \boldsymbol{w}_1\T A^{-1}\boldsymbol{w}_1\Big)}_{C_2}\,\eta^2.
\end{aligned}
\]
Define constants $C_0,C_1,C_2\in\mathbb{F}_q$ as above, we have $\Delta(\eta) = C_2 \eta^2 + C_1 \eta + C_0$, where $C_0,C_1,C_2$ and $\Delta(\eta)$ depend on the evaluation points \(\boldsymbol{\alpha}\). 
Since $C_0$ is nonzero, $\Delta$ is not the zero polynomial. 
Since \(\Delta\) is a nonzero polynomial of degree at most two and \(|\mathbb{F}_q^*|=q-1\geq 3\), one may choose \(\eta \in \mathbb{F}_q^*\) outside its set of roots. The conclusion then follows from Theorem \ref{theorem4.1}.
\end{proof}

The following  example illustrates the above theorem.

\begin{example}
\label{example4.1}
Let $t=3$, $q = 25$, $r = 2$, $n = 10$, $k = 12$, $h = 11$.
Let $z$ be a primitive element of $\mathbb{F}_{25}$, and let $g$ be a primitive element of $\mathbb{F}_{625}$.
Let $\boldsymbol{\alpha} = (g^{190}, g^{320}, g^{248}, g^{239}, g^{134}, g^{114}, g^{19}, g^{262}, g^{11}, g^{224} )$ $ \in \mathbb{F}_{625}^{10}$, then $\det(A)=z^9 \ne 0$.
Using Magma, we have $\Delta(\eta) = z^4 \eta^2 + z^{17} \eta + z^{14}$ and $\Delta(\eta)=0$ when $\eta=4$ or $\eta=z^{22}$ (here $z^{12} = -1 = 4$ in $\mathbb{F}_{25}$).
Let $\eta=z^5$, we have $\Delta(\eta) \ne 0$.
From Theorem 4.1, we have $\mathscr{T}_{12}(\boldsymbol{\alpha},\boldsymbol{1},3,11,z^5)$ is an ACD code.
\end{example}

\begin{remark}
\label{remark4.2}
For \(0\le h\le k-1\), swapping the \((h+1)\)-th row with the last row and simultaneously swapping the \((h+1)\)-th column with the last column of \(\operatorname{Tr}(GG\T)\) restores the block structure in Theorem \ref{theorem4.1}. Hence, the existing proof method may be adapted.
\end{remark}

Next, we present a construction of additive MDS ACD code from additive GTRS code.

\begin{theorem}
\label{theorem4.2}
Let $q$ be a prime power and $n,k,t$ and $r \geq 2$ be positive integers such that $n \le |\mathfrak{B}|$, $3 \leq k < nr$, $1 \leq t \leq nr-k$, $h =k-1$, $\eta \in \mathbb{F}_q^*$.
Let $k=ar+b$, $t=ur+v$, where $0 \leq b,v \leq r-1$, and let $S$ be as in Theorem \ref{theorem3.2}.
Let $\boldsymbol{\alpha}=(\alpha_1,\alpha_2,\dots,\alpha_n)\in \mathbb{F}_{q^r}^n$, where no two elements among $\alpha_1,\alpha_2,\dots,\alpha_n$ form a conjugate pair over $\mathbb{F}_q$ and each $\alpha_i$ has exactly $r$ distinct conjugates over $\mathbb{F}_q$.
Denote $T_m=\sum\limits_{i=1}^n \operatorname{Tr}(\alpha_i^m)\in \mathbb{F}_q$, $A=
\big(T_{i+j}\big)_{0\le i,j\le k-2}$ and suppose $\det(A) \neq 0$, $\boldsymbol{w} = \big(\, T_{k-1}+\eta T_{k-1+t},\; T_{k}+\eta T_{k+t},\;\dots,\; T_{2k-3}+\eta T_{2k-3+t} \,\big)\T$ and $\Delta=T_{2k-2}+2\eta T_{2k-2+t}+\eta^2 T_{2k-2+2t}-\boldsymbol{w}\T A^{-1}\boldsymbol{w}$.
If $(\eta,\boldsymbol{\alpha}) \notin \Lambda$ and $\Delta \neq 0$, then the code $\mathscr{T}_k(\boldsymbol{\alpha},\boldsymbol{1},t,k-1,\eta)$ is an additive MDS ACD code, where
\[
\Lambda =\bigcup_{s\in S}\Lambda _s=\bigcup_{s\in S}\ \bigcup_{\substack{I \subseteq \{1, \dots, n\}\\ |I| = a+s}} \left\{(\eta, \boldsymbol{\alpha}) : P_s(I) \text{ holds and } \eta^{-1} = \sum_{i_1+i_2=h}c_{i_1}f_{i_2}\in \mathbb{F}_q^*\right\},
\]
with $M=(a+s)r$ and $0 \leq i_1 \leq m$, $0 \leq i_2 \leq M$, where $M$, $m$ and \(P_s(I)\) are as in the Notation of Section 3.
When $S=\varnothing$, we have $\Lambda=\varnothing$, and the code
$\mathscr{T}_k(\boldsymbol{\alpha},\boldsymbol{1},t,k-1,\eta)$ is an additive MDS ACD code for every
$\eta\in\mathbb{F}_q^*$ with $\Delta\neq 0$.
\end{theorem}

\begin{proof}
The MDS property follows from Theorem \ref{theorem3.2}, while the ACD property follows from Theorem \ref{theorem4.1} because $\det(A)\neq 0$ and $\Delta\neq 0$.
\end{proof}

\begin{example}
\label{example4.2}
Let $t=3$, $q = 25$, $r = 2$, $n = 10$, $k = 9$ and $h = 8$. Let $z$ be a primitive element of $\mathbb{F}_{25}$, and let $g$ be a primitive element of $\mathbb{F}_{625}$.
Let $\boldsymbol{\alpha} = (g^{287}, g^{235}, g^{80}, g^{20}, g^{394}, g^{295}, g^{147}, g^{397},g^{62}, g^{58}) $ $\in \mathbb{F}_{625}^{10}$. Using Magma, we have $\Lambda  =  \Lambda _1$,
where $\Lambda _1$ denotes
\[
f_8 = f_9^2,
\]
\[
\eta^{-1} \in \big\{ f_7 - f_8 \cdot f_9 \,\big\vert\, I \subseteq \{1,2,\dots,10\},\ |I|=5 \big\}
= \big\{1, z^3 ,z^6,z^8,z^{18},z^{21},z^{23}\big\},
\]
then $\det(A)=z^{19} \ne 0$.
Using Magma, we have $\Delta(\eta) = z^{17} \eta^2 + z^{9}$ and $\Delta(\eta)=0$ when $\eta=z^2$ or $\eta=z^{14}$.
Then let $\eta = z^{12} \in \mathbb{F}_{25}$ and $\eta^{-1}=z^{12}$, i.e. $(\eta,\boldsymbol{\alpha}) \notin \Lambda $ and $\Delta(\eta) \ne 0$.
From Theorem 4.2 we know that $\mathscr{T}_9(\boldsymbol{\alpha},3,8,z^{12})$ is an additive $[10,9,6]$-code over $\mathbb{F}_{625}$, and hence it is an additive MDS ACD code.
\end{example}

\begin{remark}
\label{remark4.3}
Similar to Remark \ref{remark4.2}, by applying the same row and column permutations to the matrix, the proof method of Theorem 4.2 can be adapted for general \(h\).
\end{remark}

\section{Conclusion}
In this paper, we have constructed new additive non-RS MDS codes using twisted Reed-Solomon codes, and we have established necessary and sufficient conditions for them to be MDS.
We have also obtained a family of additive GTRS codes whose duals are still additive GTRS codes, and have established a construction of ACD codes by virtue of these GTRS codes.

\end{document}